\documentclass[a4paper, UKenglish, cleveref, autoref, thm-restate,cleverref]{lipics-v2021}

\title{Synchronized CTL over One-Counter Automata} 

\author{Shaull Almagor}{Department of Computer Science, Technion, Israel \and \url{https://shaull.cswp.cs.technion.ac.il/}} {shaull@technion.ac.il}{https://orcid.org/0000-0001-9021-1175}{supported by the Israel Science Foundation grant 989/22}

\author{Daniel Assa}{Reichman University, Herzliya, Israel}{daniel.assa@post.runi.ac.il}{}{}

\author{Udi Boker}{Reichman University, Herzliya, Israel \and \url{https://faculty.runi.ac.il/udiboker/}}{udiboker@runi.ac.il}{https://orcid.org/0000-0003-4322-8892}{supported by the Israel Science Foundation grant 2410/22}

\authorrunning{S. Almagor, D. Assa, and U. Boker} 

\Copyright{Shaull Almagor, Daniel Assa, and Udi Boker} 



\category{} 

\relatedversion{} 

\nolinenumbers 

\EventEditors{Patricia Bouyer and Srikanth Srinivasan}
\EventNoEds{2}
\EventLongTitle{43rd IARCS Annual Conference on Foundations of Software Technology and Theoretical Computer Science (FSTTCS 2023)}
\EventShortTitle{FSTTCS 2023}
\EventAcronym{FSTTCS}
\EventYear{2023}
\EventDate{December 18--20, 2023}
\EventLocation{IIIT Hyderabad, Telangana, India}
\EventLogo{}
\SeriesVolume{284}
\ArticleNo{15}

\usepackage{hyperref}
\usepackage{mathtools}
\usepackage[dvipsnames,table,xcdraw]{xcolor}
\usepackage{xspace}
\usepackage{complexity}
\usepackage{tikz} 
\usetikzlibrary{automata, positioning, arrows, decorations.pathreplacing,patterns}
\usepackage[disable]{todonotes}
\usepackage{enumitem}

\newcommand{\tuple}[1]{\langle #1  \rangle}

\newcommand{\True}{\mathtt{true}}
\newcommand{\False}{\mathtt{false}}

\newcommand{\lin}{\mathrm{Lin}}

\newcommand{\cA}{{\mathcal A}}

\newcommand{\cK}{{\mathcal K}}

\newcommand{\Nat}{\ensuremath{\mathbb{N}}\xspace}

\newcommand{\bbN}{\mathbb{N}}

\newcommand{\leads}[1]{\stackrel{#1}{\leadsto}}
\renewcommand{\vec}[1]{\boldsymbol{#1}}

\renewcommand{\phi}{\varphi}
\newcommand{\limplies}{\to} 

\newcommand{\CondFont}[1]{\texttt{#1}}
\newcommand{\eqZ}{\CondFont{=0}}

\newcommand{\gZ}{\CondFont{>0}}

\newcommand{\ZTs}{zero transitions\xspace}

\newcommand{\effect}[1]{\emph{effect}(#1)}

\newcommand{\shift}{\textsf{shift}} %
\newcommand{\Prev}[1]{{#1_{prev}}} %

\newcommand{\core}{\texttt{core}} 

\newcommand{\Const}[1]{\texttt{#1}} 
\newcommand{\cT}{\Const{cT}} 
\newcommand{\sT}{\Const{sT}} 
\renewcommand{\P}{\Const{P}} 

\newcommand{\period}{\mathtt{p}}
\newcommand{\thresh}{\mathtt{t}}
\newcommand{\lcm}{\mathrm{lcm}}

\newcommand{\Subject}[1]{\paragraph*{#1.}}

\newcommand{\CS}{CTL+Sync\xspace}
\newcommand{\CUA}{CTL+UA\xspace}

\newcommand{\IA}[1]{\cA^{#1}}

\newcommand{\es}{\Const{s}} 
\newcommand{\segStart}{\mathsection}

\newcommand{\complP}{{\mathsf{P}}}
\newcommand{\complNP}{{\mathsf{NP}}}
\newcommand{\complPSPACE}{{\mathsf{PSPACE}}}

\newcommand{\FullVersion}[1]{#1}
\newcommand{\CameraReady}[1]{}

\crefname{enumi}{}{}

\begin{document}
\maketitle
	
\begin{abstract}
We consider the model-checking problem of \emph{Synchronized Computation-Tree Logic} (CTL+Sync) over One-Counter Automata (OCAs).
CTL+Sync augments CTL with temporal operators that require several paths to satisfy properties in a synchronous manner, e.g., the property ``all paths should eventually see $p$ at the same time''. 
The model-checking problem for CTL+Sync over finite-state Kripke structures was shown to be in $\complP^{\complNP^\complNP}$. 
OCAs are labelled transition systems equipped with a non-negative counter that can be zero-tested. Thus, they induce infinite-state systems whose computation trees are not regular.
The model-checking problem for CTL over OCAs was shown to be $\complPSPACE$-complete. 

We show that the model-checking problem for CTL+Sync over OCAs is decidable. However, the upper bound we give is non-elementary. We therefore proceed to study the problem for a central fragment of CTL+Sync, extending CTL with operators that require \emph{all} paths to satisfy properties in a synchronous manner, and show that it is in $\mathsf{EXP}^\mathsf{NEXP}$ (and in particular in $\mathsf{EXPSPACE}$), by exhibiting a certain ``segmented periodicity'' in the computation trees of OCAs.

\end{abstract}
\newpage
\section{Introduction}
Branching-time model checking is a central avenue in formal verification, as it enables reasoning about multiple computations of the system with both an existential and universal quantification. As systems become richer, the classical paradigm of e.g., CTL model checking over finite-state systems becomes insufficient. To this end, researchers have proposed extensions both of the logics~\cite{axelsson2010extended,clarkson2010hyperproperties,clarkson2014temporal,BCHK14,CD16,ABK16} and of the systems~\cite{cook2017verifying,walukiewicz2000model,demri2010model,GL13}. 

In the systems' frontier, of particular interest are infinite-state models. Typically, such models can quickly lead to undecidability (e.g., two-counter machines~\cite{Min67}). However, some models can retain decidability while still having rich modelling power. One such model that has received a lot of attention in recent years is One Counter Automata (OCAs)~\cite{valiant1975deterministic,haase2009reachability} -- finite state machines equipped with a non-negative counter that can be zero-tested. 
Model checking CTL over OCAs was studied in~\cite{GL13}, where it was shown to be decidable in $\complPSPACE$. The main tool used there is the fact that despite the infinite configuration space, the computations of an OCA do admit some periodic behavior, which can be exploited to exhibit a small-model property for the satisfaction of Until formulas.

In the logics' frontier, a useful extension of CTL is that of CTL with Synchronization operators (\CS), introduced in~\cite{CD16}. \CS extends CTL with operators that express synchronization properties of computation trees. Specifically, two new operators are introduced: \emph{Until All} and \emph{Until Exists}. The former, denoted by $\psi_1 UA\psi_2$, holds in state $s$ if there is a uniform bound $k\in \bbN$ such that $\psi_2$ holds in all paths from $s$ after exactly $k$ steps, and $\psi_1$ holds in all paths up to step $k$. Thus, intuitively, it requires all the computations to synchronize the satisfaction of the Until operator. The latter, somewhat less natural operator, denoted by $\psi_1 UE\psi_2$, requires that there exists a uniform bound $k$ such that in every level $j<k$ of the computation tree, some path satisfies $\psi_1$ and can be continued to satisfy $\psi_2$ at level $k$. 

In comparison, the standard CTL operators $A\psi_1 U\psi_2$ and $E\psi_1 U\psi_2$ require that all paths/some path satisfy the Until formula, but there is no requirement that the bounds coincide. We illustrate the differences between the semantics in \cref{fig:CTL_Sync_Semantics}.
As discussed in~\cite{CD16}, \CS can describe non $\omega$-regular properties of trees, and hence goes beyond MSO, while retaining a decidable model-checking problem over finite Kripke structures.

\newcommand{\circBlack}{%
	\begin{tikzpicture}[inner sep=0pt,baseline=-0.8ex]%
		\fill[black] (0,0) circle [radius=4pt];
	\end{tikzpicture}%
}
\newcommand{\circWhite}{%
	\begin{tikzpicture}[inner sep=0pt,baseline=-0.8ex]%
		\fill[fill=black!10,draw=black] (0,0) circle [radius=4pt];
	\end{tikzpicture}%
}
\newcommand{\circStripes}{%
	\begin{tikzpicture}[inner sep=0pt,baseline=-0.8ex]%
		\draw[pattern=north east lines] (0,0) circle [radius=4pt];
	\end{tikzpicture}%
}

\begin{figure}[ht]
	\centering
	\begin{subfigure}{0.45\textwidth}
		\includegraphics[scale=0.5]{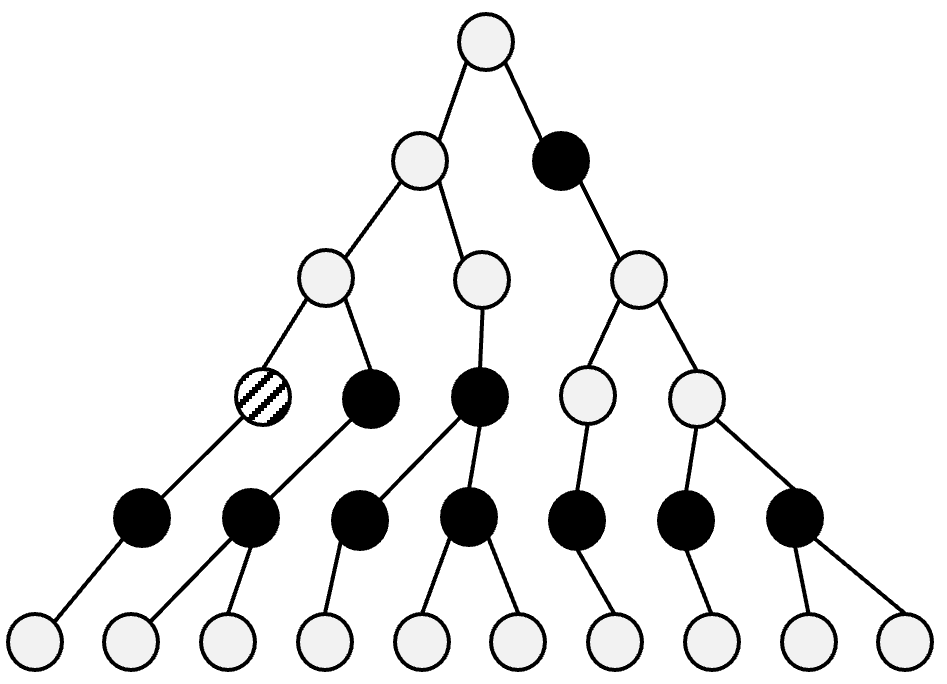}
		\caption{$\True\,U\!A\,\circBlack$ and $E\,\circWhite\,U\,\circStripes$ are satisfied.}
		\label{subfig:FA_EU}
	\end{subfigure}
	\begin{subfigure}{0.45\textwidth}     
		\includegraphics[scale=0.5]{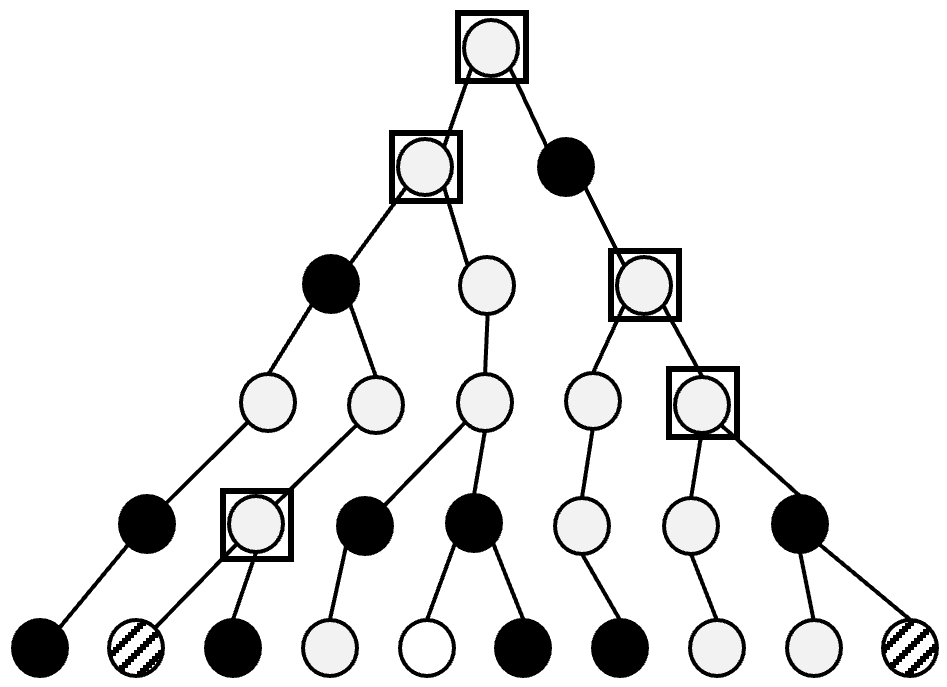}
		\caption{$A\,\True\,U\,\circBlack$ and $\circWhite\,U\!E\,\circStripes$ are satisfied.}
		\label{subfig:AF_UE}
	\end{subfigure}
	\caption{The computation tree in~\textbf{(\subref{subfig:AF_UE})} satisfies $A\,\True\,U\,\circBlack$, since every path eventually reaches $\circBlack$. However, it does not satisfy $\True\,U\!A\,\circBlack$, since this eventuality does not happen synchronously. In contrast, the tree in~\textbf{(\subref{subfig:FA_EU})} does satisfy $\True\,U\!A\,\circBlack$.\\
	The formula $E\,\circWhite\,U\,\circStripes$ is satisfied in \textbf{(\subref{subfig:FA_EU})} (along the left branch), but not in \textbf{(\subref{subfig:AF_UE})}, which only satisfies the weaker $\circWhite\,U\!E\,\circStripes$. The boxed nodes on each level show that $\circWhite$ holds along a path to $\circStripes$ at level $6$.
	}
	\label{fig:CTL_Sync_Semantics}
\end{figure}

In this work, we show the decidability of \CS model checking over OCAs: Given an OCA $\cA$ and a \CS formula $\phi$, the problem of whether $\cA$ satisfies $\phi$.
We thus combine the expressiveness of \CS with the rich modeling power of OCAs.

On the technical side, 
the approach taken in \cite{CD16} to solve model-checking of \CS over Kripke structures does not seem to be very useful in our case. The solution there is based on the observation that every two levels of the computation tree that share the same set of Kripke-states must also share the same satisfaction value to every \CS formula. Hence, in that case, the algorithm can follow the computation tree of the powerset of the given Kripke structure, and terminate when encountering a level that has the same set of states as some previous level.
In contrast, for OCAs, the unbounded counter prevents the ability to consider subsets of configurations.


On the other hand, the approach taken in \cite{GL13} to solve model-checking of OCAs with respect to CTL is indeed useful in our case. Specifically, the algorithm in~\cite{GL13} is based on an analysis of the periodic behavior of the set of counter values that satisfy a given CTL formula in a given state of the OCA.
We extend this approach, taking into account the additional complexity that stems from the synchronization requirements; see \cref{sec:DirectModelChecking}.

We start with establishing the decidability of \CS model checking using Presburger Arithmetic (see \cref{sec:Presburger}). This, however, yields a procedure with non-elementary runtime.

We then proceed to our main technical contribution (\cref{sec:DirectModelChecking}), providing an  algorithm for model checking the central fragment of \CS that extends CTL with the $UA$ operator, which requires \emph{all} paths to satisfy properties in a synchronous manner.
Its running time is in $\mathsf{EXP}^\mathsf{NEXP}$ (and in particular in $\mathsf{EXPSPACE}$), and for a fixed OCA and formulas of a fixed nesting depth, it is in $\complP^\complNP$ (and in particular in $\complPSPACE$).

 Since \CS makes assertions on the behavior of different paths at the same time step (namely the same level of the computation tree), we need to reason about which configurations occur at each level of the tree.
More precisely, in order to establish decidability we wish to exhibit a small-model property of the form \emph{if the computation tree from some configuration $(s,v)$, for a state $s$ and counter value $v$, satisfies the formula $\phi$, then the computation tree from some configuration $(s,v')$ for a small $v'$ satisfies $\phi$ as well}. 
Unfortunately, the computation trees of an OCA from two configurations $(s,v)$ and $(s,v')$ cannot be easily obtained from one another using simple pumping arguments, due to the zero tests. This is in contrast to the case where one does not care about the length of a path, as in~\cite{GL13}. 
To overcome this, we show that computation trees of an OCA $\cA$ can be split into several segments, polynomially many in the size of $\cA$, and that within each segment we can find a bounded prefix that is common to all trees after a certain counter threshold, and such that the remainder of the segment is periodic. Using this, we establish the small model property above. The toolbox used for proving this, apart from careful cycle analysis, includes 2TVASS -- a variant of 2VASS studied in~\cite{LS20} that allows for one counter of the 2VASS to be zero-tested.

We believe that this structural result (\cref{cl:Periodicity}) is of independent interest for reasoning about multiple traces in an OCA computation tree, when the length of paths plays a role.

\section{Preliminaries}\label{sec:Preliminaries} 
Let $\bbN=\{0,1,\ldots\}$ be the natural numbers. For a finite set $A\subseteq \bbN$ we denote by $\lcm(A)$ the least common multiple of the elements in $A$.

For a finite sequence $\xi=t_0 t_1 \cdots t_{h{-}1}$ and integers $x,y$, such that $0\leq x \leq y \leq h{-}1$, we write $\xi[x..y]$ for the infix of $\xi$ between positions $x$ and $y$, namely for $t_x t_{x+1} \cdots t_y$. We also use the parentheses `(' and `)'  for a non-inclusive range, e.g., $[x..y)=[x..y{-}1]$, and abbreviate $\xi[x..x]$ by $\xi[x]$.
We denote the length of $\xi$ by $|\xi|=h$.
\paragraph*{One Counter Automata}
A \emph{One Counter Automaton} (OCA) 
 is a triple $\cA=\tuple{S, \Delta, L}$, where $S$ is a finite set of states, $\Delta \subseteq (S \times \{\eqZ, \gZ \} \times \{0, +1, -1\} \times S)$ is a transition relation, 
 and $L:S\to AP$, for some finite set $AP$ of \emph{atomic propositions}, is the state labeling.

 A pair $(s,v)\in S\times\Nat$ is a \emph{configuration} of $\cA$. We write $(s,v)\to_t (s',v')$ for a transition $t\in \Delta$ if one of the following holds:
 \begin{itemize}
 	\item $t=(s,\eqZ,e,s')$, with $e\in \{0,+1\}$, $v=0$ and $v'=e$, or
 	\item $t=(s,\gZ,e,s')$, with $e\in \{0,+1,-1\}$, $v>0$ and $v'=v+e$.
 \end{itemize}
We write $(s,v)\to (s',v')$ if $(s,v)\to_t(s',v')$ for some $t\in \Delta$.

We require that $\Delta$ is total, in the sense that for every configuration $(s,v)$ we have $(s,v)\to (s',v')$ for some configuration $(s',v')$. Note that this is a syntactic requirement --- every state should have outgoing transitions on both $\eqZ$ and $\gZ$. 
This corresponds to the standard requirement of Kripke structures that there are no deadlocks. We denote by $|\cA|$ the number of states of $\cA$. Note that the description size of $\cA$ is therefore polynomial in $|\cA|$.

A \emph{path} of $\cA$ is a sequence of transitions $\tau=t_1,\ldots,t_k$ such that there exist states $s_0,\ldots,s_k$ where $t_i=(s_{i-1},\bowtie_i,e_i,s_i)$ with $\bowtie_i\in \{\eqZ,\gZ\}$ for all $1\le i\le k$. We say that $\tau$ is \emph{valid} from starting counter $v_0$ (or from configuration $(s_0,v_0)$), if there are counters $v_0,\ldots,v_k\in \bbN$ such that for all $1\le i<k$ we have 
$(s_{i-1},v_{i-1})\to_{t_i} (s_i,v_i)$.
We abuse notation and refer to the sequence of configurations also as a \emph{path}, starting in $(s_0,v_0)$ and ending in $(s_k,v_k)$. 
The \emph{length} of the path $\tau$ is $k$, and we define its $\effect{\pi}=\sum_{i=1}^k e_i$.

We also allow infinite paths, in which case there are no end configurations and the length is $\infty$. In this case we explicitly mention that the path is infinite.
We say that $\tau$ is \emph{balanced/positive/negative} if $\effect{\tau}$ is zero/positive/negative, respectively. 
It is a \emph{cycle} if $s_0=s_k$, and it \emph{has a cycle} $\beta$, if $\beta$ is a cycle and is a contiguous infix of $\tau$.

\paragraph*{\CS}
A CTL+Sync formula $\phi $ is given by the following syntax, where $q$ stands for an atomic proposition from a finite set $AP$ of atomic propositions.
\[\phi ::= \underbrace{\True \mid q \mid \phi \land \phi \mid \neg \phi \mid EX\phi \mid E\phi U \phi \mid A\phi U \phi}_{\text{Standard CTL}}
\mid \underbrace{\phi UA \phi \mid \phi UE \phi}_{\text{Sync operators}}\]

We proceed to the semantics. Consider an OCA $\cA=\tuple{S, \Delta, L}$, a configuration $(s,v)$, and a CTL+Sync formula $\phi$. Then $\IA{(s,v)}$ satisfies $\phi$, denoted by $\IA{(s,v)} \models \phi$ , as defined below.


\begin{itemize}
	\item[]\textbf{Boolean Opeartors:}
	\begin{itemize}
	\item $\IA{(s, v)} \models \True$; $\IA{(s, v)} \not\models \False$\qquad  \labelitemi\ \ $\IA{(s, v)} \models q$ if $q \in L(s)$.
	\item $\IA{(s, v)} \models \neg \phi$ if $\IA{(s, v)} \not\models \phi$. 
	\qquad ~~~
	\labelitemi\ \ $\IA{(s, v)} \models \phi \land \psi$ if $\IA{(s, v)} \models \phi$ and $\IA{(s, v)} \models \psi$ .
	\end{itemize}
	\item[]\textbf{CTL temporal operators:}
	\begin{itemize}
	\item $\IA{(s, v)} \models EX\phi$ if $(s, v) \to (s', v')$ for some configuration $(s', v')$ such that $\IA{(s', v')} \models \phi$.
	\item $\IA{(s, v)} \models E\phi U\psi$ if there exists 
	a valid path $\tau$ from $(s,v)$
	and $k \geq 0$, such that $\IA{\tau[k]} \models \psi$ and for every $j\in[0..k-1]$, we have $\IA{\tau[j]}  \models \phi$.
	\item $\IA{(s, v)} \models A\phi U\psi$ if for every 
	valid path $\tau$ from $(s,v)$
	there exists $k \geq 0$, such that $\IA{\tau[k]} \models \psi$ and for every $j\in[0..k-1]$, we have $\IA{\tau[j]}  \models \phi$.
	\end{itemize}
	\item[]\textbf{Synchronization operators:}
	\begin{itemize}
	\item $\IA{(s, v)} \models \phi UA \psi$ if there exists $k \geq 0$, such that
	for every valid path $\tau$ from $(s,v)$ of length $k$ and for every $j\in[0..k-1]$ we have $\IA{\tau[j]}  \models \phi$  and $\IA{\tau[k]} \models \psi$.
	
	\item $\IA{(s, v)} \models \phi UE \psi$ if there exists $k \geq 0$, such that for every $j\in[0..k-1]$ there exists a valid path $\tau$ from $(s,v)$ of length $k$
	such that $\IA{\tau[j]} \models \phi$  and $\IA{\tau[k]} \models \psi$.
	
	\end{itemize}
\end{itemize}
\begin{remark}[Additional operators]
	Standard additional Boolean and CTL operators, e.g., $\lor, EF, EG$, etc. can be expressed by means of the given syntax. Similar shorthands can be defined for the synchronization operators, e.g., $FE$ and $GE$, etc. 
	We remark that one can also consider operators such as $XE \psi$ with the semantics ``in the next step there exists a path satisfying $\psi$''. However, the semantics of this coincides with the CTL operator $EX$.
\end{remark}

\paragraph*{Presburger Arithmetic}
Presburger Arithmetic (PA)~\cite{presburger1929uber} is the first-order theory $\textrm{Th}(\bbN,0,1,+,<,=)$ of $\bbN$ with addition and order. We briefly survey the results we need about PA, and refer the reader to~\cite{haase2018survival} for a detailed survey.

For our purposes, a PA formula $\varphi(x_1,\ldots,x_d)$, where $x_1,\ldots, x_d$ are free variables, is evaluated over $\bbN^d$, and \emph{defines} the set $\{(a_1,\ldots,a_d)\in \bbN^d\mid (a_1,\ldots,a_d)\models \varphi(x_1,\ldots,x_d)\}$. It is known that PA is decidable in 2-\NEXP~\cite{cooper1972theorem,fischer1998super}.

A \emph{linear set} is a set of the form $\lin(B,P)=\{\vec{b}+\lambda_1 \vec{p_1}+\ldots \lambda_k\vec{p_k}\mid \vec{b}\in B,\ \lambda_1\ldots,\lambda_k\in \bbN\}$ where $B\subseteq \bbN^d$ is a finite \emph{basis} and $P=\{\vec{p_1},\ldots ,\vec{p_k}\}\subseteq \bbN^d$ are the \emph{periods}. 
A \emph{semilinear set} is then a finite union of linear sets.
A fundamental result about PA~\cite{ginsburg1964bounded} is that the sets definable in PA are exactly the semilinear sets, and moreover, one can effectively obtain from a PA formula $\phi$ a description of the semilinear set satisfying a formula $\phi$, and vice versa.

\begin{observation}
\label{obs:PA dim 1}
In dimension $1$, semilinear sets are finite unions of arithmetic progressions. By taking the $\lcm$ of the (nonzero) periods of the progressions and modifying the basis accordingly, we can assume a uniform nonzero period, and an additional finite set (corresponding to the zero period). 
That is, a semilinear set $S\subseteq \bbN$ is $\lin(B,\{p\})\cup C$ for effectively computable finite sets $B,C\subseteq \bbN$ and $p\in \bbN$. 
\end{observation}

\section{Periodicity and Flatness over OCAs}
\label{sec:periodicity_in_CTL_Sync}
Recall that the configuration space of an OCA is $S\times \bbN$. The underlying approach we take to solve \CS model checking is to show that satisfaction of \CS formulas over these configurations exhibits some periodicity. Moreover, the run tree of the OCA can be captured, to an extent, using a small number of cycles (a property called \emph{flatness}). These properties will be relied upon in proving our main results.

\subsection{Periodicity}
In this subsection we formalize our notions of ultimate periodicity, show how they suffice for model-checking, and cite important results about periodicity in CTL.

Consider a CTL+Sync formula $\phi$. We say that $\phi$ is \emph{$(\thresh(\phi),\period(\phi))$-periodic} with respect to an OCA $\cA$ (or just \emph{periodic}, if we do not care about the constants) if for every state $s\in S$ and counters $v,v'>\thresh(\phi)$, if $v\equiv v'\mod \period(\phi)$ then $(s,v)\models \phi\iff (s,v')\models \phi$. We think of $\thresh(\phi)$ as its threshold and of $\period(\phi)$ as its period. 
We say that $\phi$ is \emph{totally $(\thresh(\phi),\period(\phi))$-periodic} with respect to $\cA$ if every subformula of $\phi$ (including $\phi$ itself) is $(\thresh(\phi),\period(\phi))$-periodic with respect to $\cA$.
We usually omit $\cA$, as it is clear from context.

Total-periodicity is tantamount to periodicity for each subformula, in the following sense.
\begin{proposition}\label{cl:TotalPeriodicity}
	\label{cl:totally sat iff sat in every level}
	A CTL+Sync formula $\phi$ is totally $(\thresh(\phi),\period(\phi))$-periodic if and only if every subformula $\psi$ of $\phi$ is $(\thresh'(\psi),\period'(\psi))$-periodic for some constants $\thresh'(\psi),\period'(\psi)$.
\end{proposition}

For a totally periodic formula, model checking over OCA can be reduced to model checking over finite Kripke structures, as follows. Intuitively, we simply ``unfold'' the OCA and identify states with high counter values according to their modulo of $\period(\phi)$.
\begin{proposition}\label{cl:OCAasKripke}
	\label{cl:model checking periodic formulas}
	Consider an OCA $\cA$ and a totally $(\thresh(\phi),\period(\phi))$-periodic CTL+Sync formula $\phi$, then we can effectively construct a Kripke structure $\cK$ of size $|\cA|\cdot(\thresh(\phi)+\period(\phi))$ such that $\cA\models \phi$ if and only if $\cK\models \phi$.
\end{proposition}

In~\cite[Theorem 1]{GL13} it is proved that every CTL formula $\phi$ over OCA is periodic. Our goal is to give a similar result for \CS, which in particular contains CTL. In order to avoid replicating the proof in~\cite{GL13} for CTL, we observe that the proof therein is by structural induction over $\phi$, and that moreover -- the inductive assumption requires only periodicity of the subformulas of $\phi$. We can thus restate~\cite[Theorem 1]{GL13} with the explicit inductive assumption, so that we can directly plug our results about \CS into it.

Denote by $k=|S|$ the number of states of $\cA$, and let $K=\lcm(\{1,\ldots, k\})$. 
\begin{theorem}[Theorem 1 in~\cite{GL13}, restated]
	\label{cl:SL13_Theorem1}
	Consider a CTL+Sync formula $\phi$ whose outermost operator is a CTL operator, and whose subformulas are periodic, then $\phi$ is periodic, and we have the following.
	\begin{itemize}[topsep=0pt]
		\item If $\phi=\True$, $\phi=\False$, or $\phi=p$ for $p\in AP$, then $\phi$ is $(0,1)$-periodic.
		\item If $\phi=\neg \psi$ then $\thresh(\phi)=\thresh(\psi)$ and $\period(\phi)=\period(\psi)$.
		\item If $\phi=\psi_1\wedge \psi_2$ then $\thresh(\phi)=\max\{\thresh(\psi_1),\thresh(\psi_2)\}$ and $\period(\phi)=\lcm(\period(\psi_1),\period(\psi_2))$.
		\item If $\phi=EX\psi$ then $\thresh(\phi)=\thresh(\psi)+\period(\psi)$ and $\period(\phi)=K\cdot \period(\psi)$
		\item If $\phi=E \psi_1 U \psi_2$ or $\phi=A \psi_1 U \psi_2$ then\footnote{Note that in~\cite{GL13} the case of $A\psi_1 U \psi_2$ is not stated, but rather the dual Release operator $E \psi_1 R \psi_2$, which follows the same proof.} $\thresh(\phi)=\max\{\thresh(\psi_1),\thresh(\psi_2)\}+2\cdot k^2 \cdot \lcm(K\cdot \period(\psi_1),\period(\psi_2))$ and $\period(\phi)=\lcm(K\cdot \period(\psi_1),\period(\psi_2))$.
	\end{itemize}
\end{theorem}

\subsection{Linear Path Schemes}
\label{sec:LPS}
The runs of an OCA can take intricate shapes. Fortunately, however, we can use the results of~\cite{LS20} about \emph{flatness} of a variant of 2-VASS with some zero tests, referred to as 2-TVASS, to obtain a simple form that characterizes reachability, namely linear path schemes.

A \emph{linear path scheme} (LPS) is an expression of the form $\pi=\alpha_0\beta_1^*\alpha_1\cdots \beta_k^*\alpha_k$ where each $\alpha_i\in \Delta^*$ is a path in $\cA$ and each $\beta_i\in \Delta^*$ is a cycle in $\cA$. The \emph{flat length} of $\pi$ is  $|\pi|=|\alpha_0\beta_1\alpha_1\cdots \beta_k\alpha_k|$, the \emph{size} of $\pi$ is $k$.

A concrete path $\tau$ in $\cA$ is \emph{$\pi$-shaped} if there exist $e_1,\ldots,e_k$ 
such that $\tau=\alpha_0\beta_1^{e_1}\alpha_1\cdots \beta_k^{e_k}\alpha_k$.

Our first step is to use a result of~\cite{LS20} on 2-TVASS to show that paths of a fixed length in $\cA$ admit a short LPS.
The idea is to transform the OCA $\cA$ to a 2-TVASS $\cA'$ by introducing a length-counting component. That is, in every transition of $\cA'$ as a 2-TVASS, the second component increments by 1. 
\begin{lemma}\label{cl:OCAbyTVASS}
	\label{cl:short_LPS}
	Let $(s,v)$ and $(s',v')$ be configurations of $\cA$. If there exists a path $\tau$ of length $\ell$ from $(s,v)$ to $(s',v')$, then there is also a $\pi$-shaped path $\tau'$ of length $\ell$ from $(s,v)$ to $(s',v')$, where $\pi$ is some linear path scheme of flat length $\text{poly}(|S|)$ and size $O(|S|^3)$.
\end{lemma}

By~\cref{cl:short_LPS}, we can focus our attention to $\pi$-shaped paths where $\pi$ is ``short''. Henceforth, we call a path $\tau$ \emph{basic} if it is $\pi$-shaped for some LPS $\pi$ as per~\cref{cl:short_LPS}.

Using standard acceleration techniques (see e.g.,~\cite{leroux2005flat,finkel2018reachability,LS20}), we also get from~\cref{cl:short_LPS} that the reachability relation of an OCA (including path length) is effectively semilinear. More precisely, we have the following.
\begin{corollary}
	\label{cl:OCA to PA}
	We can effectively compute, for every $s,s'\in S$, a PA formula $Path_{s,s'}(x,y,x',y')$ such that $(v,\ell,v',\ell')\models Path_{s,s'}(x,y,x',y')$ if and only if a path of length $\ell$ ending\footnote{It is more natural to assume $\ell=0$ and simply consider paths starting at $(s,v)$. However, our formulation makes things easier later on.} at $(s,v)$ can be continued to a path of length $\ell'$ ending at $(s',v')$.
\end{corollary}

\section{Model Checking CTL+Sync via Presburger Arithmetic}\label{sec:Presburger}
In this section we show that model checking CTL+Sync over OCAs is decidable, by reducing it to the satisfiability problem of a PA formula. 


We start with a simple observation.
\begin{lemma}
	\label{cl:periodic_to_Presburger}
	Consider a totally $(\thresh(\phi),\period(\phi))$-periodic CTL+Sync formula $\phi$, then for every state $s\in S$ we can compute a PA formula $P_{\phi,s}(x)$ such $v\models P_{\phi,s}(x)$ if and only if  $(s,v)\models \phi$.
\end{lemma}
Next, we show that from a PA formula as above we can obtain a threshold and a period.
\begin{lemma}
	\label{cl:Presburger_to_periodic}
	Consider a CTL+Sync formula $\phi$ and PA formulas $P_{\phi,s}(x)$, for every state $s$, such that $v\models P_{\phi,s}(x)$ iff $(s,v)\models \phi$. Then $\phi$ is $(\thresh(\phi),\period(\phi))$-periodic for computable constants $\thresh(\phi)$ and $\period(\phi)$.
\end{lemma}

Combining~\cref{cl:SL13_Theorem1,cl:periodic_to_Presburger} we obtain that every CTL formula (without Sync) can be translated to PA formulas. We now turn to include the Sync operators. 

Consider a CTL+Sync formula $\phi$. We construct, by induction on the structure of $\phi$, PA formulas $P_{s,\phi}(v)$, for every state $s\in S$, such that $P_{s,\phi}(v)$ holds if and only if $(s,v)\models \phi$. For the Sync operators, this utilizes the PA formulas of~\cref{cl:OCA to PA}.
\begin{itemize}
	\item If $\phi=p$ for an atomic proposition $p$, then $P_{s,\phi}(v)=\text{True}$ if $s$ is labeled with $p$ and $\text{False}$ otherwise.
	\item If $\phi=\neg\psi$, then $P_{s,\phi}(v) = \neg P_{s,\psi}(v)$.
	\item If $\phi=\psi_1 \land \psi_2$, then $P_{s,\phi}(v) = P_{s,\psi_1}(v) \land P_{s,\psi_2}(v)$.
	\item If $\phi=\text{EX}\psi$, $\phi=\text{A}\psi_1\text{U}\psi_2$ or $\phi=\text{E}\psi_1\text{U}\psi_2$ then by the induction hypothesis, $\psi$, $\psi_1$ and $\psi_2$ have corresponding PA formulas, and by \cref{cl:Presburger_to_periodic} we can compute $\thresh(\psi)$ and $\period(\psi)$ (and similarly for $\psi_1$ and $\psi_2$) such that $\psi$ is $(\thresh(\psi),\period(\psi))$-periodic. Then, by \cref{cl:SL13_Theorem1} we can compute $\thresh(\phi),\period(\phi)$ such that $\phi$ is $(\thresh(\phi),\period(\phi))$-periodic. Thus, we can apply \cref{cl:periodic_to_Presburger} to obtain PA formulas for $\phi$.
	\item If $\phi=\psi_1 \text{UE} \psi_2$, then $P_{s,\phi}(v) = \exists \ell. 
	\forall \ell'<\ell. \biggl(\bigvee_{s'\in S} \Bigl(\exists v'. \bigl(Path_{s,s'}(v,0,v',\ell') \land  P_{s',\psi_1}(v') 
	\land
	 \bigvee_{s''\in S} \exists v''. Path_{s',s''}(v',\ell',v'',\ell) \land  P_{s'',\psi_2}(v'')\bigr)\Bigr)\biggr)$
	\item If $\phi=\psi_1 \text{UA} \psi_2$, then $P_{s,\phi}(v) = \exists \ell. \biggl(\Bigl(\exists v'. \bigvee_{s'\in S}  \bigl(Path_{s,s'}(v,0,v',\ell) \land P_{s',\psi_2}(v')\bigr)\Bigr)
	~\land 		
	\Bigl(\bigwedge_{s'\in S} \forall v'. \bigl(Path_{s,s'}(v,0,v',\ell) \limplies P_{s',\psi_2}(v')\bigr)\Bigr)
	~\land $\\$
	\forall \ell'<\ell. \Bigl(\bigwedge_{s'\in S} \forall v'. \bigl(Path_{s,s'}(v,0,v',\ell') \limplies P_{s',\psi_1}(v')\bigr)\Bigr)\biggr)$
\end{itemize}
The semantics of the obtained PA formulas match the semantics of the respective Sync operators. By \cref{cl:Presburger_to_periodic} we now have that for every CTL+Sync formula $\phi$ we can compute $\thresh(\phi),\period(\phi)$ such that $\phi$ is $(\thresh(\phi),\period(\phi))$-periodic. By~\cref{cl:totally sat iff sat in every level} we can further assume that $\phi$ is totally $(\thresh(\phi),\period(\phi))$-periodic. Finally, by~\cref{cl:model checking periodic formulas} we can decide whether $\cA\models \phi$.

\begin{remark}[Complexity]
	\label{rmk:complexity_PA}
	Observe that the complexity of our decision procedure via PA formulas is non-elementary. Indeed, when using a CTL subformula, we translate it, using \cref{cl:periodic_to_Presburger}, to a PA formula that may be exponential in the size of the formula and of the OCA. Thus, we might incur an exponential blowup in every step of the recursive construction, leading to a tower of exponents. 
	\end{remark}

\section{Model Checking the \CUA Fragment}
\label{sec:DirectModelChecking}
In this section we consider the fragment of \CS induced by augmenting CTL with only the Sync operator $\psi_1 UA \psi_2$. For this fragment, we are able to obtain a much better upper bound for model checking, via careful analysis of the run tree of an OCA.

Throughout this section, we fix an OCA $\cA=\tuple{S, \Delta, L}$ with $n=|S|\geq 3$ states, and a \CUA formula $\phi$. 
Consider a configuration $(s,v)$ of $\cA$ and a \CUA formula $\phi=\psi_1 UA \psi_2$ with $UA$ being the outermost operator. The satisfaction of $\phi$ from $(s,v)$ is determined by the computation tree of $\cA$ from $(s,v)$. Specifically, we have that $\IA{(s,v)}\models \phi$ if there exists some bound $k\in \bbN$ such that $\psi_2$ holds in all configurations of level $k$ of the computation tree, and $\psi_1$ holds in all configurations of levels $\ell$ for $0\le \ell<k$.  

Therefore, in order to reason about the satisfaction of $\phi$, it is enough to know which configurations appear in each level of the computation tree. This is in contrast with UE, where we would also need to consider the paths themselves. Fortunately, it means we can use the LPS of~\cref{cl:short_LPS} to simplify the proofs.

\begin{theorem}\label{cl:PeriodicityUA}
	Given an OCA $\cA$ with $n=|\cA|$ and a CTL+UA formula $\phi$, we can compute a counter threshold $\cT$ and a counter period $\P$, both single exponential in $n$ and in the nesting depth of $\phi$, such that $\phi$ is $(\cT,\P)$-periodic with respect to $\cA$.	
\end{theorem}
Before we delve into the proof of~\cref{cl:PeriodicityUA}, we show how it implies our main result.
\begin{theorem}
	\label{cl:CTL_UA_Mode_checking_is_decidable}
	The model-checking problem for \CUA is decidable in $\textsc{EXP}^{\textsc{NEXP}}$.
\end{theorem}
\begin{proof}
	Consider a \CUA formula $\phi$ and an OCA $\cA$. By~\cref{cl:PeriodicityUA} we can compute $\cT,\P$ single exponential in $|\phi|$ and $|\cA|$, such that $\phi$ is $(\cT,\P)$-periodic. We then apply~\cref{cl:model checking periodic formulas} to reduce the problem to model checking $\phi$ against a Kripke structure $\cK$ of size $|\cA|\cdot (\cT+\cP)$. 
	
	Finally, the proof of~\cite[Theorem 1]{CD16} shows that model checking the \CUA fragment can be done in $\complP^\complNP$ in the size of the Kripke structure and the formula, yielding an $\textsc{EXP}^{\textsc{NEXP}}$ bound in our setting.
\end{proof}
\begin{remark}[Lower Bounds]
	The \PSPACE-hardness of model checking CTL over OCA from~\cite{GL13} immediately implies \PSPACE-hardness in our setting as well. However, tightening the gap between the lower and upper bounds remains an open problem. 
	\end{remark}
The remainder of the paper is devoted to proving~\cref{cl:PeriodicityUA}.

\Subject{Cycle Manipulation and Slope Manipulation}
A fundamental part of our proof involves delicately pumping and removing cycles to achieve specific counter values and/or lengths of paths. We do this with the following technical tools.

For a path $\tau$, we define the \emph{slope} of $\tau$ as $\frac{\effect{\tau}}{|\tau|}$.
Recall that a basic path is of the form $\tau=\alpha_0\beta_1^{e_1}\alpha_1\cdots\beta_k^{e_k}\alpha_k$ adhering to some LPS, where $k=O(|\cA|)$ and each $\alpha_i$ and $\beta_i$ is of length $poly(|\cA|)$. 
We denote by $b$ the maximum flat length of any LPS for a basic path. 
In particular, $b$ bounds the flat length of the LPS, the size of it, and the length of any cycle or path in it.

We call a cycle \emph{basic} if it is of length at most $b$.
A slope of a path is \emph{basic} if it may be the slope of a basic cycle, namely if it equals $\frac{x}{y}$, where $x\in[-b..b], y\in[1..b]$ and $|x|\leq y$. We denote the basic slopes by $\es_i$, starting with $\es_1$ for the smallest.
For example for $b = 3$, the basic slopes are $(\es_1\!\!=\!\!-1,\, \es_2\!\!=\!\!-\frac{2}{3},\, \es_3\!\!=\!\!-\frac{1}{2},\, \es_4\!\!=\!\!-\frac{1}{3},\, \es_5\!\!=\!\!0,\, \es_6\!\!=\!\!\frac{1}{3},\, \es_7\!\!=\!\!\frac{1}{2},\, \es_8\!\!=\!\!\frac{2}{3},\, \es_9\!\!=\!\!1)$.
Observe that for every $i$, we have $|\es_i|\ge \frac1b$, and for every $j>i$, we have $\es_j-\es_i\ge \frac{1}{b^2}$ and when they are both negative, also $\frac{\es_j}{\es_i}\leq1-\frac{1}{b^2}$.

\begin{proposition} \label{cl:CycleCombinationForElr}
	Consider a basic path $\tau$ and basic cycles $c_1, c_2, c_3$ in $\tau$ with effects $e_1, e_2, e_3\in[-b..b]$, respectively, and lengths $\ell_1, \ell_2, \ell_3\in[1..b]$, respectively, such that $\frac{e_1}{\ell_1} \leq \frac{e_2}{\ell_2} \leq \frac{e_3}{\ell_3}$. Then there are numbers $k_1,k_3\in[0..b^2]$, such that the combination of $k_1$ repetitions of $c_1$ and $k_3$ repetitions of $c_3$ yield an effect and length whose ratio is $\frac{e_2}{\ell_2}$.
\end{proposition}

\begin{proposition} \label{cl:CycleCombinationForLength}
	Consider a path $\pi$ with cycles $c_1, c_2$ with effects $e_1, e_2\in[-b..b]$ and lengths $\ell_1, \ell_2\in[1..b]$, respectively, such that $\frac{e_1}{\ell_1} < \frac{e_2}{\ell_2}$, and a length $x$ that is divisible by $\lcm[1..2b^2]$. Then there are numbers $k_1,k_2\in[0..b\cdot x]$, such that the addition or removal of $k_1$ repetitions of $c_1$ and the addition or removal of $k_2$ repetitions of $c_2$ yield a path of the same effect as $\pi$ and of a length shorter or longer, as desired, by $x$ (provided enough cycle repetitions exist).
\end{proposition}

In order to prove~\cref{cl:PeriodicityUA}, we show that every computation tree of $\cA$, starting from a big enough counter value $v>\cT$, has a `segmented periodic' structure with respect to $\phi$. 
That is, we can divide its levels into 
$poly(n)$
many \emph{segments}, such that only the first $\sT\in Exp(n,|\phi|$) levels in each segment are in the `core', while every other level $\ell$ is a sort-of repetition of the level $\ell-\P$, for a \emph{period} $\cP$.
We further show that there is a similarity between the cores of computation trees starting with counter values $v$ and $v+\P$.
We depict the segmentation in~\cref{fig:constants}, and formalize it as follows.

Consider a formula $\phi = \psi_1 UA\psi_2$, where $\psi_1$ and $\psi_2$ are $(\cT(\psi_1),\P(\psi_1))$- and $(\cT(\psi_2),\P(\psi_2))$-periodic, respectively. We define several constants to use throughout the proof.
\begin{itemize}[topsep=0pt]
	\item[] \hspace{-18pt} \textbf{Constants depending only on the number $n$ of states in the OCA}
	\item $b\in Poly(n)$: the bound on the length of a linear path scheme on $\cA$.
	\item $B=\lcm[1..2b^3]$.
%
	\item[] \hspace{-18pt} \textbf{Constants depending on $n$ and the CTL+UA formula $\phi$}
	\item $\Prev{\P}(\phi)=\lcm(\P(\psi_1),\P(\psi_2))$ -- the unified period of the subformulas.
	\item $\Prev{\cT}(\phi)=\max(\cT(\psi_1),\cT(\psi_2))$ -- the unified counter threshold of the subformulas.
	\item  $\P(\phi)=B\cdot\Prev{\P}(\phi)$ -- the `period' of $\phi$.
	\item $\sT(\phi)=b^9 \cdot \P(\phi)$ -- the `segment threshold' of $\phi$.
	\item  $\cT(\phi)= b^{11}\cdot \P(\phi)$ -- the `counter threshold' of $\phi$.
\end{itemize}
Eventually, these periodicity constants are plugged into the inductive cases of~\cref{cl:SL13_Theorem1}, as shown in the proof of \cref{cl:PeriodicityUA}. Then, all the constants are single-exponential in $n$ and the nesting depth of $\phi$.
Notice that the following relationship holds between the constants.
\begin{proposition}\label{cl:PeriodBiggerThanPreviousThreshold}
	$\P(\phi)>\Prev{\cT}(\phi)$ for $b\ge 2$. 
\end{proposition}

When clear from the context, we omit the parameter $\phi$ from $\P, \sT, \cT, \Prev{\P},\Prev{\cT}$.

We provide below an intuitive explanation for the choice of 
constants above.

\begin{figure}[ht]
	\includegraphics[scale=.85]{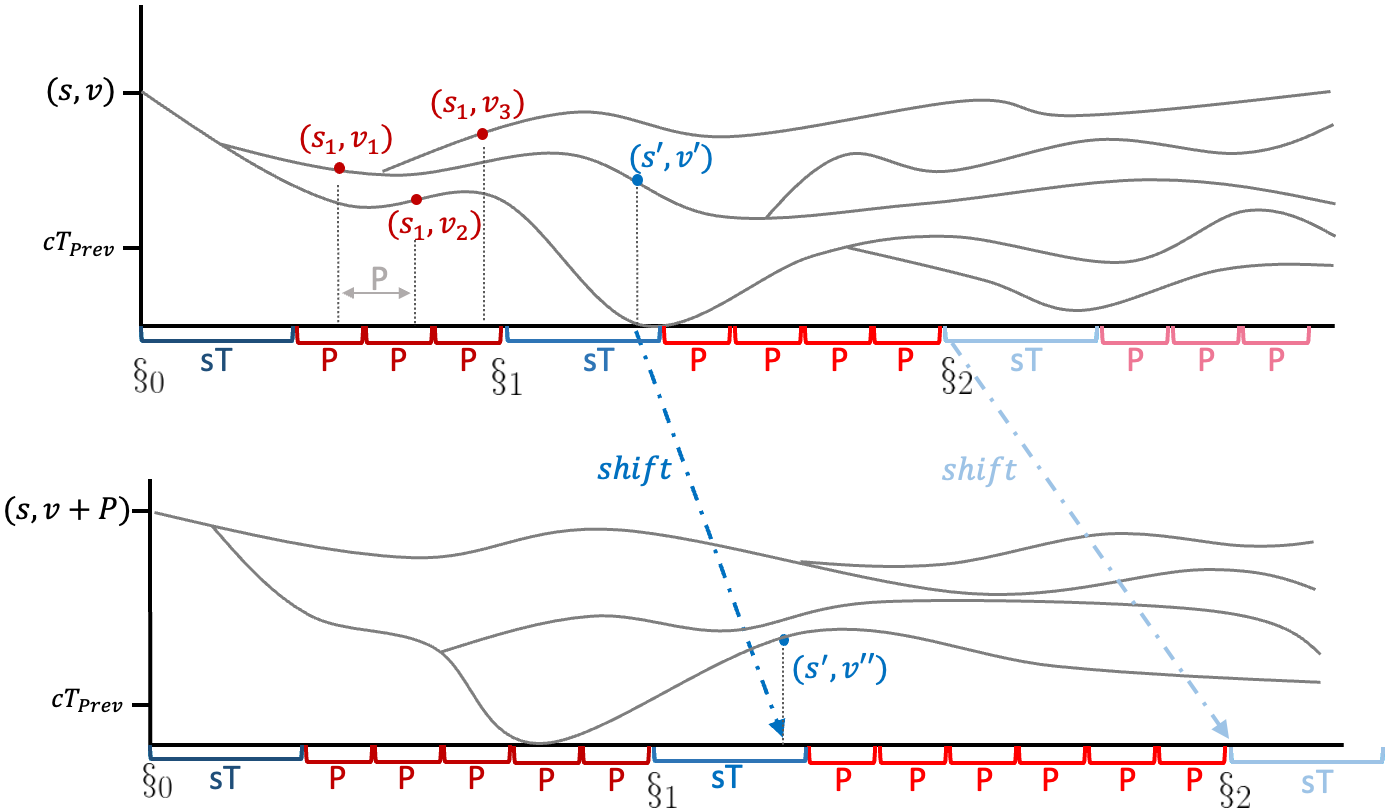}
	\caption{A computation tree from $(s,v)$ above and from $(s,v+\P)$ below, demonstrating the \emph{segmented periodicity}. Each segment $i$ starts at $\segStart_i$ and ends at $\segStart_{i+1}-1$ (except for the last, which never ends).
	The \emph{core} of the computation tree is the union of the first $\sT$ positions from each segment. Following these $\sT$ positions, there is within each segment periodicity of length $\P$, meaning that for every path of length $\ell$, there exists an equivalent path of length $\ell+\P$, as shown for the three points with the same state $s$ and equivalent counter values $v_1\equiv v_2\equiv v_3$.
	Between the trees, there is equivalence between the core positions, mapped via the $\shift$ function, as shown for the two points with the same state $s'$ and the equivalent counter values $v'\equiv v''$.
	}
	\label{fig:constants}
\end{figure}

\Subject{Intuition for the period $\P$}
The period has two different roles: \emph{level-periodicity} within each segment of a computation tree, and \emph{counter-value periodicity} between two computation trees starting with different counter values.

\emph{Level periodicity within a segment}: For lengthening or shortening a basic path by $\P$, we add and/or remove some copies of its cycles. Adding or removing $\Prev{\P}$ copies of the same cycle guarantees that the end counter values of the original and new paths are equivalent modulo $\Prev{\P}$. Since the cycles in a basic path are of length in $[0..b]$, setting $\P$ to be divisible by $\Prev{\P}\cdot \lcm[1..b]$, allows to add or remove $\frac{\P}{|c|}$ copies of a cycle $c$, where $\frac{\P}{|c|}$ is divisible by $\Prev{\P}$, as desired. Yet, we might need to add copies of one cycle and remove copies of another, thus, as per \cref{cl:CycleCombinationForLength}, we need $\P$ to be divisible by $\Prev{\P}\cdot lcm[1..2b^2]$.

\emph{Counter periodicity between computation trees}: We change a path $\tau$ that starts with a counter value $v$ to a path that starts with a counter value $v+\P$, or vice versa, by lengthening or shortening it by $\frac{\P}{\es}$, respectively, where $\es$ is a positive basic slope. In some cases, we need to also make sure that the longer or shorter path has a drop bigger or smaller, respectively, than $\tau$ by exactly $\P$.

As $\frac{\P}{\es}$ is bounded by $b\cdot\P$, if there are at least $b\cdot\P$ repetitions of a cycle $c$ in $\tau$ whose slope is $-\es$, we can just add or remove $\frac{b\cdot\P}{|c|}$ copies of $c$, so we need $\P$ to be divisible by $\Prev{\P}\cdot \lcm[1..b]$, for guaranteeing that the counter values at the end of the original and new paths are equivalent.
Yet, in some cases we need to combine two cycles, as per \cref{cl:CycleCombinationForElr}. As the combination of the two cycles might be of length up to $2b^3$, we need $\P$ to be divisible by $\Prev{\P}\cdot \lcm[1..2b^3]$.

\Subject{Intuition for the counter threshold $\cT$ and the segment threshold $\sT$}
In order to apply~\cref{cl:CycleCombinationForElr,cl:CycleCombinationForLength}, we need to have in the handled path many repetitions of two cycles of different slopes. We thus choose $\cT$ and $\sT$ to be large enough so that paths in which only one (negative) cycle slope is repeated many times must hit zero within a special region called the `core' of the tree, as defined below.



\Subject{The core of a computation tree}
For every counter value ${v}>\cT$, the `core' with respect to a fixed formula $\phi$, denoted by $\core({v})\subseteq \Nat$, of a computation tree of $\cA$ that starts with a counter value ${v}$ consists of $m+1<b^2$ segments, with each segment corresponding to a negative basic slope and having $\sT$ consequent numbers. 
For every $i\in[0..m]$, the start of Segment $i$ depends on the initial counter value $v$ of the computation tree, it is denoted by $\segStart_i(v)$, and it is defined as follows:
\begin{itemize}
	\item $\segStart_0(v)=0$,
	\item For $i\in[1..m]$, we set $\segStart_{i}(v)= \frac{-1}{\es_{i}}(v-\Prev{\cT})-b^8\cdot \P$.
	\item For convenience, we also define $\segStart_{m{+}1}(v)=\infty$.
\end{itemize}
Then, we define $\core(v)=\bigcup_{i=0}^m[\segStart_i(v)..(\segStart_i(v)+\sT-1)]$.

Observe that the core of every tree is an ordered list of $((m+1)\cdot \sT)$ numbers (levels), while just the starting level of every segment depends on the initial counter value $v$.
We can thus define a bijection $\shift:\core(v)\to\core(v+\P)$ that maps the $i$-th number in $\core({v})$ to the $i$-th number in $\core({v+\P})$ (see also~\cref{fig:constants}).

Recall that we define $\sT$ so that, intuitively, if a path is long enough to reach $\segStart_i(v)+\sT$ without reaching counter value $\Prev{\cT}$, then the path must have many cycles with a slope larger than $\es_i$, and if the path manages to reach $\Prev{\cT}$ before $\segStart_{i+1}(v)$ (namely the end of Segment $i$), then it must have many cycles with a slope at least as small as $\es_i$. This is formalized as follows.

\begin{lemma} \label{cl:SegmentsCycles}
	Let $\tau$ be a basic path, or a prefix of it, of length $\ell$ starting from counter value $v>\cT$, and let $i\in[0..m]$.
	\begin{enumerate}
		\item \label{Item:SlowCycles} 
		If $\ell\ge \segStart_i(v)+\frac{\sT}{2}$ and the counter values of $\tau$ stay above $\Prev{cT}$, then $\tau$ has a cycle with slope $\es_j$ for $j> i$ that repeats at least $b^4\cdot \P$ times. 
		\item \label{Item:FastCycles} 
		If $\ell< \segStart_{i+1}(v)$ and the counter values of $\tau$ reach  $\Prev{cT}$, then $\tau$ has a cycle with slope $\es_j$ for $j\leq i$ that repeats at least $b^4\cdot \P$ times. 
	\end{enumerate}
\end{lemma}
As a sanity check, the lemma states that if a path $\tau$ reaches $\Prev{\cT}$ for the first time at length $\ell\in[\segStart_1(v){+}\frac{\sT}{2}~..~ \segStart_2(v))$, then it has many cycles with slope $\es_1=-1$ (enough to decrease down to $\Prev{\cT}$ before $\segStart_2(v)$), as well as many cycle with slope at least $\es_2=\frac{-(b-1)}{b}$ (enough to keep above $\Prev{\cT}$ through $\segStart_1(v)+\frac{\sT}{2}$).
Observe also that a path cannot reach $\Prev{\cT}$ at Segment $0$, namely before  $\segStart_1(v)$.

\Subject{The segment and $\shift$ periodicity}
Consider a threshold $T$ and period $P$. We say that counter values $u,v$ are \emph{$(T,P)$-equivalent}, denoted by $u\equiv_{T,P}v$ if either $u,v\ge T$ and $P$ divides $|u-v|$, or $u,v<T$ and $u=v$. That is, either both $u,v$ are greater than $T$, in which case they are equivalent modulo $P$, or they are both smaller than $T$ and are equal.

The segment periodicity within a computation tree is then stated as Claim \cref{item:SegmentPeriodicity} in \cref{cl:Periodicity} below, while the similarity between computation trees starting from counters $v$ and $v+\P$ as Claim \cref{item:CoreShift}. (By $(s,v)\leads{\ell}(s',v')$ we mean that the computation tree starting with state $s$ and counter value $v$ has a path of length $\ell$ ending in state $s'$ and counter value $v'$.)

\begin{lemma}\label{cl:Periodicity}
	Consider states $s$ and $e$, a counter value $v>\cT$, an arbitrary counter value $u$, and an arbitrary path length $\ell$.
	\begin{enumerate}
		\item \label{item:SegmentPeriodicity} 
		If $\ell\not\in\core(v)$ then: 
		\begin{enumerate}
			\item \label{item:SegmentPeriodicityShortening} 
			$(s,v)\leads{\ell}(e,u) \implies (s,v)\leads{\ell-\P}(e,u')$, and 
			\item \label{item:SegmentPeriodicityLengthening} 
			$(s,v)\leads{\ell-\P}(e,u) \implies  (s,v)\leads{\ell}(e,\tilde{u})$,
		\end{enumerate}
		for  some counter values $u'$ and $\tilde{u}$, such that $u\equiv_{\Prev{\cT},\Prev{\P}} u'\equiv_{\Prev{\cT},\Prev{\P}}\tilde{u}$.
		\item\label{item:CoreShift}
		 If $\ell\in\core(v)$ then: 
		 \begin{enumerate}
			\item \label{item:CoreShiftShortening}
			$(s,v+\P)\leads{\shift(\ell)}(e,u) \implies  (s,v)\leads{\ell}(e,u')$, and
			\item \label{item:CoreShiftLengthening}
			$(s,v)\leads{\ell}(e,u) \implies (s,v+\P)\leads{\shift(\ell)}(e,\tilde{u})$  
	\end{enumerate}
		for  some counter values $u'$ and $\tilde{u}$, such that $u\equiv_{\Prev{\cT},\Prev{\P}} u'\equiv_{\Prev{\cT},\Prev{\P}}\tilde{u}$.
	\end{enumerate}
\end{lemma}

	Throughout the proof, we will abbreviate $u\equiv_{\Prev{\cT},\Prev{\P}}u'$ by  $u\equiv u'$.
	We split the proof into four parts, each devoted to one of the four stated implications. 
	In each of them, we assume the existence of a path $\tau$ that witnesses the left side of the implication, say $(s,v)\leads{\ell}(e,u)$, and show that there exists a path $\tau'$ that witnesses the right side of the implication, say $(s,v)\leads{\ell-\P}(e,u')$, where $u\equiv u'$. By \cref{cl:short_LPS}, we assume that $\tau$ is a basic path.
	We present some of the cases; the remaining parts are in the appendix. 


\subsubsection*{Proof of \cref{cl:Periodicity}.\cref{item:SegmentPeriodicityShortening}}

Let $\tau$ be a basic path of length $\ell\notin \core(v)$ such that $(s,v)\leads{\ell}(e,u)$ via $\tau$.
We construct from $\tau$ a path $\tau'$ for $(s,v)\leads{\ell-\P}(e,u')$, such that $u\equiv u'$.
The proof is divided to two cases. 

\paragraph*{Case 1a.1: The counter values in $\tau$ stay above $\Prev{\cT}$}
If there is no position in $\tau$ with counter value $\Prev{\cT}$, then in particular $\tau$ has no zero-transitions.
Since $\ell\notin \core(v)$, then in particular $\ell\ge \sT>3b^5\P$.
Thus, there are at least $3b^4\P$ cycle repetitions in $\tau$. 

If there is a non-positive cycle $c$ that is repeated at least $\P$ times, we can obtain $\tau'$ by removing $\frac{\P}{|c|}$ copies of $c$, as the counter values along $\tau'$ are at least as high as the corresponding ones in $\tau$.
Observe that $\tau'$ is of length $\ell-\P$ from $(s,v)$ to $(e,u')$ with $u'=u-\effect{c}\frac{\P}{|c|}$. Since we have $u'\ge u\ge \Prev{\cT}$, then $u\equiv u'$.

Otherwise, each non-positive cycle in $\tau$ is taken at most $\P$ times. Thus, the positive cycles are repeated at least $3b^4\P-b\P  \ge 3b^3\P$ times.
In particular, there exists a positive cycle $c$ that repeats at least $3b^2\P$ times. By removing $\frac{\P}{|c|}$ occurrences of it, we obtain a path $\tau'$ of length $\ell-\P$. Notice first that this path is valid. Indeed, up until the cycle $c$ is taken, the path $\tau'$ coincides with $\tau$, so the counter remains above $\Prev{\cT}$. 
Since $c$ is a positive cycle, after completing its iterations, the counter value becomes at least $3b^2\P-\P+\Prev{\cT}$. Then, even if all remaining transitions in the negative cycles have effect $-1$, the counter value is reduced by at most $b^2\P$ (as there are at most $(b-1)\P$ remaining cycles, each of effect at least $-b$, and the simple paths in $\tau$ can reduce by another $b$ at most).
Thus, the value of the counter remains at least $3b^2\P-\P+\Prev{\cT} - b^2\P > \Prev{\cT}$.
Finally, let $(e,u')$ be the configuration reached at the end of $\tau'$, then $u-u'=\effect{c} \frac{\P}{|c|}$, so $u\equiv u'$.

\paragraph*{Case 1a.2: $\tau$ reaches counter value $\Prev{\cT}$.}

Let $0\le z_f\le z_u\le \ell$ be the first and ultimate positions in $\tau$ where the counter value is exactly $\Prev{\cT}$. 
We split $\tau$ into three parts: $\tau_1 = \tau[0~..~z_f), \tau_2=\tau[z_f~..~z_u), \tau_3=\tau[z_u~..~\ell]$ (it could be that $z_f=z_u$, in which case the middle part is empty).
Since $\tau$ is of length $\ell\ge \sT\geq b^9\cdot\P$, then at least one of the parts above is of length at least $b^8\cdot\P$ (recall $b\ge 3$). We split according to which part that is. For simplicity, we start with the cases that $\tau_2$ or $\tau_3$ are long, and only then handle the case of a long $\tau_1$.
\begin{enumerate}
	\item {\it The middle part  $\tau_2=\tau[z_f~..~z_u]$ is of length at least $b^8\cP$.} \label{item:LongMiddlePart}
	
	As $\tau_2$ is of length at least $b^8\cdot\P$, some cycle $c$ in it must repeat at least $b^6\cdot\P$ times.
	
	If $c$ is balanced, we can obtain $\tau'$ by removing $\frac{\P}{|c|}$ of its repetitions.
	
	If $c$ is positive, starting at position $x$ with counter value $v_x$, then the counter value at position $y$ where $c$'s repetitions end is at least $v_x+b^6\cdot\P$. As $\tau_2$ eventually gets down to $\Prev{\cT}<\P$, there must be a negative cycle $c_{-}$ that repeats at least $b\cdot\P$ times between position $y$ and the first position after $y$ that has the counter value $v_x+\frac{b^6\cdot\P}{2}$. Hence, we can obtain $\tau'$ by removing repetitions of $c$ and $c_{-}$, as per \cref{cl:CycleCombinationForLength}, ensuring that the only affected values are above $v_x$.
	
	If $c$ is negative, starting at position $x$ with counter value $v_x$, then $v_x\geq b^6\cdot\P$.
	As $\tau_2$ starts with counter value $\Prev{\cT}$, and $\Prev{\cT}<P$ (\cref{cl:PeriodBiggerThanPreviousThreshold}), there must be a positive cycle $c_{+}$
	that repeats at least $b\cdot\P$ times between the last position with counter value $\frac{b^6\cdot\P}{2}$ and $x$. Hence, we can obtain $\tau'$ by removing repetitions of $c$ and $c_{+}$, as per \cref{cl:CycleCombinationForLength}, ensuring that the only affected values are above $b^5\cdot\P$.
	
	\item {\it The last part $\tau_3=\tau[z_u~..~\ell]$ is of length at least $b^8\cP$.}\label{item:LongLastPart}
	
	As in the previous case, a cycle $c$ must repeat in $\tau_3$ at least $b^6\cdot\P$ times.
	If $c$ is balanced, we can remove $\frac{\P}{|c|}$ of its repetitions, getting the desired path $\tau'$.
	
	Otherwise, it must be that $\tau_3$ stays above $\Prev{\cT}$, and reaches a value at least $b^6\cdot\P$. Indeed, if $c$ is positive then its repetitions end at some position $x$ with a counter value at least that high, and if it is negative it starts at some position $x$ with a counter value at least that high.
	
	If the counter value also drops to $\frac{b^6\cdot\P}{2}$ after position $x$, then we can remove positive and negative cycles exactly as in the previous case. Otherwise, we can just remove $\frac{\P}{|c|}$ repetitions of $c$, guaranteed that the counter value at the end of $\tau_3$ is above $\Prev{\cT}$.
	
	\item {\it Only the first part $\tau_1=\tau[0~..~z_f)$ is of length at least $b^8\cP$.}\label{item:LongFirstPart}
	
	If any of the other parts is long, we shorten them.
	Otherwise, their combined length is less than $2b^8\cdot\P<\frac{\sT}{2}$, implying that the first part $\tau_1$ is longer than $\segStart_i(v)+\frac{\sT}{2}$.
	
	Hence, by \cref{cl:SegmentsCycles}, there are `fast' and `slow' cycles $c_f$ and $c_s$, respectively, of slopes $\es_f<\es_s$, such that each of them repeats at least $b^4\cdot\P$ times in $\tau_1$.
	
	Thus, by \cref{cl:CycleCombinationForLength}, we can add and/or remove some repetitions of $c_f$ and $c_s$, such that $\tau_1$ is shorten by exactly $\P$.
	Yet, we should ensure that the resulting path $\tau'$ is valid, in the sense that its corresponding first part $\tau'_1$ cannot get the counter value to $0$.
	We show it by cases:
	\begin{itemize}
		\item If $c_f$ or $c_s$ are balanced cycles, then we can remove the balanced cycle only, without changing the remaining counter values.
		\item If there is a positive cycle $c_{+}$ that repeats at least $2b^2\cdot\P$ times, then the counter value climbs by at least $2b^2\cdot\P$ from its value $v_x$ at position $x$ where $c_{+}$ starts and the position $y$ where its repetitions end. 
		As the counter gets down to $\Prev{\cT}$ at the end of $\tau_1$, and $\Prev{\cT}<P$ (\cref{cl:PeriodBiggerThanPreviousThreshold}), there must be a negative cycle $c_{-}$ that repeats at least $b\cdot\P$ times between position $y$ and the first position after $y$ that has the counter value $v_x+b^2\cdot\P$. Hence, we can remove repetitions of $c_{+}$ and $c_{-}$, as per \cref{cl:CycleCombinationForLength}, ensuring that the only affected values are above $v_x$.
		\item Otherwise, both $c_f$ and $c_s$ are negative, implying that we add some repetitions of $c_f$ and remove some repetitions of $c_s$. We further split into two subcases:
		\begin{itemize}
			\item If $c_s$ appears before $c_f$ then there is no problem, as the only change of values will be their increase, and all the values were nonzero to begin with (as we are before $z_f$). 
			\item If $c_f$ appears first, ending at some position $x$, while $c_s$ starts at some later position $y$, then a-priori it might be that
			repeating $c_f$ up to $b\cdot\P$ more times, as per \cref{cl:CycleCombinationForLength}, will take the counter value to $0$. 
			
			However, observe that since there are at most $b-2$ positive cycles, and each of them can repeat at most $2b^2\cdot\P-1$ times, the counter value $v_x$ at position $x$, and along the way until position $y$, is at least $v_y-(b-2)2b^2\cdot\P$, where $v_y$ is the counter value at position $y$. As $c_s$ repeats at least $b^4\cdot\P$ times, we have $v_y\geq b^4\cdot\P$. Thus $v_x \geq b^4\cdot\P - 2(b-2)b^2\cdot\P > b^2\P$. Hence, repeating $c_f$ up to $b\cdot\P$ more times at position $x$ cannot take the counter value to $0$, until position $y$, as required.
		\end{itemize}
		
	\end{itemize}
	
\end{enumerate}



\subsubsection*{Proof of \cref{cl:Periodicity}.\cref{item:CoreShiftShortening}}
\Subject{The case of Segment $\segStart_{0}$}
For a path of length $\ell$, we have in Segment $0$ that $\shift(\ell)=\ell$, and indeed a path from $v$ is valid from $v+\P$ and vise versa, as they do not hit $\Prev{\cT}$: Their maximal drop is $\sT$, while $v\geq\cT>\P+\sT>\Prev{\cT} + \sT$.

We turn to the $i$th segment, for $i\geq 1$.
Consider a basic path $\tau$ for $(s,v+\P)\leads{\shift(\ell)}(e,u)$. Recall that  $\shift(\ell)=\ell+\frac{\P}{-\es_i} \in [\segStart_i(v+\P) ~..~ \segStart_i(v+\P)+\sT)]$. 
We construct from $\tau$ a path $\tau'$ for  $(s,v)\leads{\ell}(e,u')$, such that $u\equiv u'$, along the following cases.

\paragraph*{Case 2a.1: The counter values in $\tau$ stay above $\Prev{\cT}$}

As there is no position in $\tau$ with counter value $\Prev{\cT}$, then in particular $\tau$ has no zero-transitions.
We further split into two subcases:

\begin{enumerate}
	\item If $\tau$ does not have $b\cdot\P$ repetitions of a `relatively fast' cycle with slope $\es_j$ for $j\leq i$, then the
	drop of $\tau$, and of every prefix of it, is at most $X+Y$, where $X$ stands for the drop outside `slow' cycles of slope $\es_h$ for $h>i$, and $Y$ for the rest of the drop. We have $X<b^3\cdot\P$ and $Y<(\ell+\frac{\P}{\es_i})(-\es_{i+1})$.
	
	We claim that we can obtain $\tau'$ by removing $\frac{\P}{-\es_i \cdot |c|}$ repetitions of any cycle $c$, which repeats enough in $\tau$, having that the drop $D$ of $\tau'$ is less than $v-\Prev{\cT}$.
	
	Indeed, the maximal such drop $D$ might be the result of removing only cycles of slope $(+1)$, whose total effect is $\frac{\P}{-\es_i}$, having 
	$
	D\leq \frac{\P}{-\es_i} + X + Y =\frac{\P}{-\es_i} + b^3\cdot\P + (\ell+\frac{\P}{-\es_i})(-\es_{i+1}) \leq 
	b \cdot\P + b^3\cdot\P +(\ell+\frac{\P}{-\es_i})(-\es_{i+1}).
	$	
	Since $\ell<\segStart_{i+1}(v+\P)= \frac{1}{-\es_{i+1}}(v +\P - \Prev{\cT})- b^8\cdot\P$, we have
		$D \leq 
		(b^3+b)\cdot\P + (\frac{1}{-\es_{i+1}}(v +\P - \Prev{\cT})- b^8\cdot\P)+\frac{\P}{-\es_i})(-\es_{i+1}) 
		=$\\
		$
		(b^3+b)\cdot\P + v +\P - \Prev{\cT}- (-\es_{i+1})\cdot b^8\cdot\P+\frac{(-\es_{i+1}) }{-\es_i}\cdot\P)
		=
		(b^3+b+1+\frac{(-\es_{i+1})}{-\es_i})\cdot\P + v - \Prev{\cT}- (-\es_{i+1})\cdot b^8\cdot\P)
		<
		b^4 \cdot\P - b^7\cdot\P+ v - \Prev{\cT}.$
	It is thus left to show that $b^4 \cdot\P < b^7\cdot\P$, which obviously holds.
	
	\item Otherwise, namely when $\tau$ does have $b\cdot\P$ repetitions of a `relatively fast' cycle with slope $\es_j$ for $j\leq i$, let $c$ be the first such cycle in $\tau$. We can obtain $\tau'$ by removing $\frac{\P}{-\es_i \cdot |c|}$ repetitions of $c$: The counter value in $\tau'$, which starts with counter value $v$, at the position after the repetitions of $c$ will be at least as high as the counter value in $\tau$, which starts with counter value $v+\P$, after the repetitions of $c$. Notice that the counter value cannot hit $\Prev{\cT}$ before arriving to the repetitions of $c$ by the argument of the previous subcase.
	
\end{enumerate}

\paragraph*{Case 2a.2: $\tau$ reaches counter value $\Prev{\cT}$}

Again let $\tau_1 = \tau[0~..~z_f), \tau_2=\tau[z_f~..~z_u), \tau_3=\tau[z_u~..~\shift(\ell)]$ as in \cref{item:SegmentPeriodicityShortening}.

In order to handle possible \ZTs, we shorten $\tau_1$, such that the resulting first part $\tau'_1$ of $\tau'$, which starts with counter value $v$, also ends with counter value exactly $\Prev{\cT}$. 
Since $\tau_1$ reaches $\Prev{\cT}$ and is shorter than $\segStart_{i+1}(v+\P)$, it has by \cref{cl:SegmentsCycles}.\cref{Item:FastCycles} at least $b^4\cdot\P$ repetitions of a `fast' cycle of slope $\es_f\leq \es_i$. Let $c_f$ be the first such cycle. We split to cases.

\begin{enumerate}
	\item
	If $\es_f=\es_i$ or $\tau_2$ or $\tau_3$ are of length at least $b^5\cdot\P$, we can remove $\frac{\P}{-\es_f \cdot |c_f|}$ repetitions of $c_f$ in $\tau_1$. Note that the resulting first part $\tau'_1$ of $\tau'$ indeed ends with counter value $\Prev{\cT}$. However, while when  $\es_f=\es_i$  the resulting length of $\tau'$ will be $\ell$, as required, in the case that $\es_f<\es_i$, we have that $\tau'$ will be longer than $\ell$. Nevertheless, in this case, as $\tau_2$ or $\tau_3$ are of length at least $b^5\cdot\P$, we can further shorten $\tau_2$ or $\tau_3$ without changing their effect, as per \cref{cl:CycleCombinationForLength}, analogously to \cref{item:LongMiddlePart} or \cref{item:LongLastPart}, respectively, in the proof of \cref{cl:Periodicity}.\cref{item:SegmentPeriodicityShortening}.2.
	
	\item \label{Item:FastSlowInPartOne2a}
	Otherwise, we are in the case that $\tau_1$ has a `really fast' cycle of slope $\es_f< \es_i$ that repeats at least $b^4\cdot\P$ times, and both $\tau_2$ or $\tau_3$ are of length less than $b^5\cdot\P$.
	We claim that in this case $\tau_1$ must also have $b^4\cdot\P$ repetitions of a `relatively slow' cycle $c_s$ of slope $\es_s\geq \es_i$. 
	
	Indeed, assume toward contradiction that $\tau_1$ has less than $b^4\cdot\P$ repetitions of a cycle with slope $\es_s$ for $s\geq i$. Then the longest such path has less than $b$ transitions of $(+1)$ out of cycles, $b^6\cdot\P$ such transitions in cycles, and the rest of it consists of `fast' cycles with slope indexed lower than $i$. 
	
	Thus its length is at most $X + L$, where $X=b^6\cdot\P$ is the $(+1)$ transitions, and $L$ is the longest length to drop from counter value $v+\P+X$ to $\Prev{\cT}$ with `fast' cycles. Thus, $L\leq \frac{1}{-\es_{i-1}}(v+\P+X - \Prev{\cT})$.
	
	Now, we have that the length of $\tau$ is at least $\segStart_{i}(v+\P)=\frac{1}{-\es_{i}}(v - \Prev{\cT})- b^8\cdot\P$.
	Thus, parts $\tau_2$ and $\tau_3$ of $\tau$ are of length at least $Z = \frac{1}{-\es_{i}}(v - \Prev{\cT})- b^8\cdot\P - \frac{1}{-\es_{i-1}}(v+\P+X - \Prev{\cT}) - X = (\frac{1}{-\es_i} - \frac{1}{-\es_{i+1}})(v - \Prev{\cT}) - b^8\cdot\P - (1+ \frac{1}{-\es_{i+1}})b^6\cdot\P$. 
	
	Since $(\frac{1}{-\es_i} - \frac{1}{-\es_{i+1}})\geq \frac{1}{b^2}$, $(1+ \frac{1}{-\es_{i+1}}\leq b+1)$, and $v - \Prev{\cT} \geq \cT - \Prev{\cT} > \cT -\P > 3b^{10} \P$, we have $Z\geq \frac{1}{b^2}(3b^{10} \cdot\P) - b^8\cdot\P - (b+1)b^6\cdot\P = 2 b^8\cdot\P - (b+1)b^6\cdot\P$.
	Therefore, at least one of $\tau_2$ and $\tau_3$ is of length at least $b^7 \cdot \P$, leading to a contradiction.
	
	So, we are in the case that $\tau_1$ has at least $b^4\cdot\P$ repetitions of a `really fast' cycle $c_f$ of slope $\es_f< \es_i$ as well as $b^4\cdot\P$ repetitions of a `relatively slow' cycle $c_s$ of slope $\es_s\geq \es_i$. 
	
	By analyzing the different possible orders of $c_s$ and $c_f$, we can cut and repeat the cycles far enough from 0 so as to construct valid paths. See the Appendix for details.\qed
\end{enumerate}

\subsection{Proof of \cref{cl:PeriodicityUA}}
The proof is by induction over the structure of $\phi$, where 
\cref{cl:SL13_Theorem1} already provides the periodicity for all CTL operators.

It remains to plug $UA$ into the induction by showing (1)
the $(\cT, \P)$-periodicity of a formula $\phi=\psi_1 UA \psi_2$ with respect to an OCA $\cA$, provided that its subformulas are $(\Prev{\cT},\Prev{\P})$-periodic; and (2) by showing that the period $\P$ and threshold $\cT$ are single-exponential in $n=|\cA|$ and in the nesting depth of $\phi$.

\begin{enumerate}
	\item We show that for every state $s\in S$ and counters $v,v'>\cT$, if $v\equiv v'\mod \P$ then $(s,v)\models \phi\iff (s,v')\models \phi$. Withot loss of generality, write $v' = v + z\cdot\P$, for some $z\in\Nat$.
	
	\labelitemi\ If $(s,v)\models \phi$ then by the definition of the $AU$ operator and the completeness of $\cA$ we have i) there is a level $\ell$, such that for every state $e$ and counter value $u$, if $(s,v)\leads{\ell}(e,u)$ then $(e,u)\models\psi_2$, and ii) for every level $m<\ell$, state $h$ and counter value $x$, if $(s,v)\leads{m}(h,x)$ then $(h,x)\models\psi_1$.
	
	Observe first that if $\ell\not\in\core(v)$, then there also exists a level $\hat{\ell}<\ell$ witnessing $(s,v)\models \phi$, such that $\hat{\ell}\in\core(v)$. Indeed, we obtain $\hat{\ell}$, by choosing the largest level $\hat{\ell}$ in the $\core$ of $\ell$'s segment, such that $\ell\equiv\hat{\ell} \mod \P$.		
	As $\hat{\ell}<\ell$, it directly follows that for every level $\hat{m}<\hat{\ell}$, state $\hat{h}$ and counter value $\hat{x}$, if $(s,v)\leads{\hat{m}}(\hat{h},\hat{x})$ then $(\hat{h},\hat{x})\models\psi_1$.
	Now, assume toward contradiction that there is a state $\hat{e}$ and a counter value $\hat{u}$, such that  $(s,v)\leads{\hat{\ell}}(\hat{e},\hat{u})$ and $(\hat{e},\hat{u})\not\models\psi_2$. Then by (possibly several applications of) \cref{cl:Periodicity}.\cref{item:SegmentPeriodicityLengthening}, there is also a counter value $\hat{\hat{u}}\equiv_{\Prev{\cT},\Prev{\P}} \hat{u}$, such that $(s,v)\leads{\ell}(\hat{e},\hat{\hat{u}})$. As $\psi_2$ is $(\Prev{\cT},\Prev{\P})$-periodic, we have $(\hat{e},\hat{\hat{u}})\not\models\psi_2$, leading to a contradiction.
	
	Next, we claim that the level ${\ell'}=\shift^z(\hat{\ell})$, obtained by $z$ (recall that $z=\frac{v'-v}{\P}$)  applications of the $\shift$ function on $\hat{\ell}$, witnesses $(s,v')\models \phi$, namely that i) for every state $e'$ and counter value $u'$, if $(s,v')\leads{\ell'}(e',u')$ then $(e',u')\models\psi_2$, and ii) for every level $m'<\ell'$, state $h'$ and counter value $x'$, if $(s,v')\leads{m'}(h',x')$ then $(h',x')\models\psi_1$.
	
	Indeed, i) were it the case that $(e',u')\not\models\psi_2$ then by ($z$ applications of) \cref{cl:Periodicity}.\cref{item:CoreShiftShortening}, there was also a counter value $u''\equiv_{\Prev{\cT},\Prev{\P}} u'$, such that $(s,v)\leads{\hat{\ell}}(e',u'')$ and therefore $(e',u'')\not\models\psi_2$, leading to a contradiction; and ii) were it the case that $(s,v')\leads{m'}(h',x')$ and $(h',x')\not\models\psi_1$ then a) by \cref{cl:Periodicity}.\cref{item:SegmentPeriodicityShortening}, as in the argument above, there is also a level $\tilde{m}\leq m'$, such that $\tilde{m}$ is in the core of $m'$'s segment and $(s,v')\leads{\tilde{m}}(h',\tilde{x})$ where $\tilde{x} \equiv_{\Prev{\cT},\Prev{\P}} x$, and b) by \cref{cl:Periodicity}.\cref{item:CoreShiftShortening} there is a level $\hat{m}<\hat{\ell}$, such that $(s,v)\leads{\hat{m}}(h',\hat{x})$ where $\hat{x} \equiv_{\Prev{\cT},\Prev{\P}} x$
	and therefore $(h',\hat{x})\not\models\psi_1$, leading to a contradiction.
	
	\labelitemi\ If $(s,v')\models \phi$ then we have $(s,v)\models \phi$ by an argument analogous to the above, while using \cref{cl:Periodicity}.\cref{item:CoreShiftLengthening} instead of \cref{cl:Periodicity}.\cref{item:CoreShiftShortening}.
	
	\item The threshold $\cT$ and period $\P$ are calculated along the induction on the structure of the formula $\phi$. They start with threshold $0$ and period $1$, and their increase in each step of the induction depends on the outermost operator. 
	
	Observe first that we can take as the worst case the same increase in every step, that of the UA case, since it guarantees the others. Namely, its required threshold, based on the threshold and period of the subformulas, is bigger than the threshold required in the other cases, and its required period is divisible by the periods required for the other cases.
	
	Next, notice that both the threshold and period in the UA case only depend on the periods of the subformulas (i.e., not on their thresholds), so it is enough to show that the period is singly exponential in $n$ and the nesting depth of $\phi$.
	
	The period in the UA case is defined to be $\P(\phi)=B\cdot\lcm(\P(\psi_1),\P(\psi_2))$, where 
	$B=\lcm[1..2b^3]$, and $b$ is the bound on the length of a linear path scheme for $\cA$.
	By \cite{LS20}, $b$ is polynomial in $n$, and as $\lcm[1..2b^3]<4^{2b^3}$, we get that $B$ is singly exponential in $n$.
	
	Considering $\lcm(\P(\psi_1),\P(\psi_2)$, while in general $\lcm(x,y)$ of two numbers $x$ and $y$ might be equal to their multiplication, in our case, as both $\psi_1$ and $\psi_2$ are calculated along the induction via the same scheme above, they are both an exponent of $B$. Hence, $\lcm(\P(\psi_1),\P(\psi_2)=\max(\P(\psi_1),\P(\psi_2)$. Thus, we get that $\P(\phi)\leq B^x$, where $x$ is bounded by the nesting depth of $\phi$.	
\qed	
\end{enumerate}


\bibliography{bib}
\newpage
\appendix

\section{Appendix -- Omitted Proofs}
\subsection{Proofs for \cref{sec:periodicity_in_CTL_Sync,sec:Presburger}}

\FullVersion{
\begin{proof}[Proof of \cref{cl:TotalPeriodicity}]
	By definition, if $\phi$ is totally $(\thresh(\phi),\period(\phi))$-periodic then so is every subformula. We show the converse.
	
	Denote by $\mathsf{subs}$ the set of subformulas of $\phi$. Assume that every  $\psi\in \mathsf{subs}$ is $(\thresh(\psi),\period(\psi))$-periodic for some constants $\thresh'(\psi),\period'(\psi)$. Let $\thresh(\phi)=\max\{\thresh'(\psi)\mid \psi\in \mathsf{subs}\}$ and $\period(\phi')=\lcm(\{\period(\psi)\mid \psi\in \mathsf{subs}\})$. We claim that $\phi$ is totally $(\thresh(\phi),\period(\phi))$-periodic. 
	Indeed, this follows from the simple observation that if $\psi$ is $(\thresh'(\psi),\period'(\psi))$-periodic, then it is also $(\thresh''(\psi),\period''(\psi))$-periodic for every $\thresh''(\psi)>\thresh'(\psi)$ and $\period''(\psi)$ that is a multiple of $\period'(\psi)$.
\end{proof}

\begin{proof}[Proof of \cref{cl:OCAasKripke}]
	Formally, let $\cA=\tuple{S,\Delta,L}$. We define $\cK=\tuple{S\times [0..\thresh(\phi)+\period(\phi)-1],R,L'}$ where the transitions in $R$ are induced by the configuration reachability relation of $\cA$ as defined in~\cref{sec:Preliminaries} where we identify $(s,v)$ for $v\ge \thresh(\phi)+\period(\phi)$ with $\thresh(\phi)+(v\mod \period(\phi))$.
	
	We claim that  $\cA\models \phi$ if and only if $\cK\models \phi$ if and only if $\cK\models \phi$. Indeed, consider an infinite computation $\pi$ of $\cA$ and its corresponding computation $\pi'$ in $\cK$ obtained by taking modulo $\period(\phi)$ as defined above, then the set of subformulas of $\phi$ that are satisfied in each step of $\pi$ and of $\pi'$ is identical, by total $(\thresh(\phi),\period(\phi))$-periodicity. In particular, $\pi\models\phi$ if and only if $\pi'\models \phi$.
	
	Since every computation of $\cA$ induces a computation of $\cK$ (by taking modulo $\period(\phi)$ where relevant), and every computation of $\cK$ induces a computation of $\cA$ (by following the counter updates, relying on  total $(\thresh(\phi),\period(\phi))$-periodicity to make sure $0$-transitions in $\cK$ are consistent with ones in $\cA$), the equivalence follows.
\end{proof}

} 

\begin{proof}[Proof of \cref{cl:OCAbyTVASS}]
	In~\cite{LS20}, the authors consider a model called 2-TVASS, comprising a 2-dimensional vector addition system where the first counter can also be tested for zero. In our setting, this corresponds to an OCA equipped with an additional counter that cannot be tested for zero. 
	A configuration of a 2-TVASS is $(s,x_1,x_2)$ describing the state and the values of the two counters. 
	
	We transform the OCA $\cA$ to a 2-TVASS $\cA'$ by introducing a length-counting component. That is, in every transition of $\cA'$ as a 2-TVASS, the second component increments by 1. We thus have the following: there is a path of length $\ell$ from $(s,v)$ to $(s',v')$ in $\cA$ if and only if there is a path from $(s,v,0)$ to $(s',v',\ell)$ in $\cA'$.
	
	From~\cite[Corollary 16]{LS20}, using the fact that $\cA'$ has only weights in $\{-1,0,1\}$, we have that if there is a path from $(s,x_1,x_2)$ to $(s',y_1,y_2)$ in $\cA'$, then there is also such a $\pi$-shaped path where $\pi$ is of flat length $\text{poly}(|S|)$ and size $O(|S|^3)$, which concludes the proof. 
\end{proof}

\begin{proof}[Proof of \cref{cl:periodic_to_Presburger}]
	Consider the set $V_{\text{init}}=\{v \mid v\le \thresh(\phi) \text{ and } (s,v)\models \phi  \}$. We define
	\[
	P_{\phi,s}(x):= \bigvee_{u\in V_{\text{init}}}\exists m.\ x=u+m\cdot\period(\phi)  
	\]
	The correctness of this formula is immediate from the definition of $V_{\text{init}},\thresh(\phi)$, and $\period(\phi)$.
	In order to compute $P_{\phi,s}(x)$, we need to compute $V_{\text{init}}$. This is done 
	using \cref{cl:model checking periodic formulas}, by evaluating
	$\phi$ from every state of the Kripke structure.
\end{proof}

\begin{proof}[Proof of \cref{cl:Presburger_to_periodic}]
	Recall from \cref{obs:PA dim 1} that the set of satisfying assignments for a PA formula in dimension 1 is effectively semilinear and can be written as $\lin(B,\{r\})\cup C=\{b+\lambda r\mid b\in B,\ \lambda\in \bbN\}\cup C$ for effectively computable $B,C\subseteq \bbN$ and $r\in \bbN$. 
	
	For every state $s\in S$, consider therefore the set $\lin(B_s,\{r_s\})\cup C_s$ corresponding to $P_{\phi,s}$. Define $\thresh(\phi)=\max(\{b\mid b\in B_s,\ s\in S\}\cup C_s)$ and $\period(\phi)= \lcm(\{r_s\}_{s\in S})\cdot \thresh(\phi)$. 
	
	We claim that $\phi$ is $(\thresh(\phi),\period(\phi))$-periodic. Indeed, consider $v,v'>\thresh(\phi)$ such that $v\equiv v' \mod \period(\phi)$, and let $s\in S$.  Note that $v,v'\notin C_s$ for any $s$.	
	Without loss of generality assume $v>v'$, so we can write $v-v'=K\cdot \period(\phi)$ for some $K\in \bbN$.
	
	Now, if $(s,v')\models \phi$, then $v'=b'_s+\lambda'r_s$ for some $b'_s\in B_s$ and $\lambda'\in \bbN$. Then, $v=b'_s+\lambda' r_s+K\cdot \period(\phi)=b'_s+(\lambda'+K\frac{\period(\phi)}{r_s})r_s$, and note that $r_s$ divides $\period(\phi)$, so we have $v\in \lin(B_s,\{r_s\})$ and therefore $(s,v)\models \phi$.
	
	Conversely, if $(s,v)\models \phi$ then $v=b_s+\lambda r_s$ for some $b_s\in B_s$ and $\lambda\in \bbN$. Then, $v'=b_s+\lambda r_s-K\cdot \period(\phi)=b_s+(\lambda-K\frac{\period(\phi)}{r_s})r_s$. Since $v'>\thresh(\phi)\ge b_s$, then $\lambda-K\frac{\period(\phi)}{r_S}\ge 1$. Thus, $v'\in \lin(B_s,\{r_s\})$ so $(v',s)\models \phi$.
	Hence, $\phi$ is $(\thresh(\phi),\period(\phi))$-periodic.
\end{proof}

\FullVersion{

\begin{proof}[Proof of \cref{cl:CycleCombinationForElr}]\
	
	If $\frac{e_1}{\ell_1} = \frac{e_2}{\ell_2}$ or $\frac{e_1}{\ell_3} = \frac{e_2}{\ell_2}$, we can just have $k_1=1$ and $k_3=0$ or $k_3=1$ and $k_1=0$, respectively.
	
	Otherwise, we can have $k_1=(\ell_2\cdot e_3 - \ell_3\cdot e_2)$ and $k_3=(\ell_1\cdot e_2 - \ell_2\cdot e_1)$. 	Indeed, notice that $k_1$ and $k_3$ are positive and the overall ratio of effect divided by length is:
	
	\begin{align*} 
		 \frac{(\ell_2\cdot e_3 - \ell_3\cdot e_2) e_1 + (\ell_1\cdot e_2 - \ell_2\cdot e_1) e_3}
		{(\ell_2\cdot e_3 - \ell_3\cdot e_2) \ell_1 + (\ell_1\cdot e_2 - \ell_2\cdot e_1) \ell_3}   
		&= \frac{e_1 \cdot\ell_2\cdot e_3 - e_1 \cdot \ell_3\cdot e_2 + e_3 \cdot \ell_1\cdot e_2 -  e_3 \cdot  \ell_2\cdot e_1}
		{\ell_1 \cdot \ell_2\cdot e_3 - \ell_1 \cdot \ell_3\cdot e_2  + \ell_3 \cdot \ell_1\cdot e_2 - \ell_3 \cdot \ell_2\cdot e_1}  \\[5pt] 
		 =\frac{ e_3 \cdot \ell_1\cdot e_2 - e_1 \cdot \ell_3\cdot e_2 }
		{\ell_1 \cdot \ell_2\cdot e_3  - \ell_3 \cdot \ell_2\cdot e_1}  &=~
		\frac{ e_2 (e_3 \cdot \ell_1 - e_1 \cdot \ell_3)}
		{\ell_2(\ell_1\cdot e_3  - \ell_3 \cdot e_1)} ~=~ 
		\frac{e_2}{\ell_2}
	\end{align*}
\end{proof}

\begin{proof}[Proof of \cref{cl:CycleCombinationForLength}]\	
	We provide the proof for lengthening the path. The proof for shortening it is analogous.
	We split to cases.
	\begin{enumerate}
		\item $e_1=0$ or $e_2=0$: We add $k_1=\frac{x}{\ell_1}$ repetitions of $c_1$ or $k_2=\frac{x}{\ell_2}$ repetitions of $c_2$, respectively.
		\item $e_1<0$ and $e_2>0$: We add $k_1=e_2\frac{x}{e_2\cdot\ell_1 - e_1\cdot\ell_2}$ repetitions of $c_1$ and $k_2=-e_1\frac{x}{e_2\cdot\ell_1 - e_1\cdot\ell_2}$ repetitions of $c_2$.
		\item $e_1>0$ and $e_2>0$: We add $k_1=e_2\frac{x}{e_2\cdot\ell_1 - e_1\cdot\ell_2}$ repetitions of $c_1$ and remove $k_2=e_1\frac{x}{e_2\cdot\ell_1 - e_1\cdot\ell_2}$ repetitions of $c_2$.
		\item $e_1<0$ and $e_2<0$: We remove $k_1=e_2\frac{x}{e_1\cdot\ell_2 - e_2\cdot\ell_1}$ repetitions of $c_1$ and add $k_2=e_1\frac{x}{e_1\cdot\ell_2 - e_2\cdot\ell_1}$ repetitions of $c_2$.
	\end{enumerate}
\end{proof}

\begin{proof}[Proof of \cref{cl:PeriodBiggerThanPreviousThreshold}]
	When $\phi$ is an atomic proposition, there is no $\Prev{\cT}(\phi)$.
	
	Consider a formula $\phi=\psi_1 UA \psi_2$.
	We have 
	$\Prev{\cT}(\phi)=\max(\cT(\psi_1),\cT(\psi_2)) =\max(b^{11}\cdot \P(\psi_1),b^{11}\cdot \P(\psi_2)) = b^{11}\cdot \max( \P(\psi_1), \P(\psi_2))$ and
	$\P(\phi)=B\cdot\Prev{\P} = B\cdot \lcm(\P(\psi_1),\P(\psi_2)) > B\cdot\max( \P(\psi_1), \P(\psi_2))$.
	It remains to show that $B=\lcm[1..2b^3]>b^{11}$ for $b\ge 2$. To this end, it is proved in~\cite{nair1982chebyshev} that $\lcm[1..n]\ge  2^n$ for all $n\ge 7$.
	
	Then, for $b\ge 2$ we have $2b^3\ge 16>7$, so $B\ge 2^{2b^3}$.  We now prove by induction that $2^{2b^3}\ge b^{11}$ for all $b\ge 2$. For $b=1$ we have $2^{2\cdot 2^3}=2^{16}>2^{11}$. Assume correctness for $b$, we prove for $b+1$: We have
	$
	2^{2(b+1)^3}=2^{2(b^3+3b^2+3b+1)}=2^{2b^3}\cdot 2^{2(3b^2+3b+1)}\ge 
	2^{2b^3}\cdot 2^{2\cdot 3\cdot 4}=2^{24}\cdot 2^{2b^3}
	$,
	whereas $(b+1)^{11}\le (2b)^11=2^{11}\cdot  b^{11}<2^{24}\cdot 2^{2b^3}$,
	where the last inequality follows from the induction hypothesis.
\end{proof}

} 

\begin{proof}[Proof of \cref{cl:SegmentsCycles}]
	For $i=0$, the first claim trivially holds, as it claims for  a cycle with a slope at least $\es_1$, which holds for any cycle, and since $\frac{\sT}{2}>2b^5\cdot\P$, it must repeat at least $b^4\cdot \P$ times.
	The second claim vacuously holds, as no path can reach $\Prev{\cT}$ at Segment $0$.
	
	Consider $i\in[1..m]$. We prove each of the claims separately.
	
	We show that there are at least $b^5\cdot\P$ total cycle repetitions of the required slope, and since there are at most $b$ different cycles, one of them should repeat at least $b^4\cdot\P$ times.

	\begin{enumerate}
		\item
		Consider a basic path $\tau$ of length $\ell\ge \segStart_{i}(v)+\frac{\sT}{2}$, such that the counter values of $\tau$ stay above $\Prev{\cT}$.
		Assume by way of contradiction that $\tau$ has less than $b^5\cdot \P$ cycles with $\es_j$ for $j> i$. We will show that $\ell$ must fall short of $\segStart_{i}(v)+\frac{\sT}{2}$.
		
		Since there are at most $b^5\cdot \P-1$ cycles with $\es_j$ with $j> i$, each of length at most $b$, then the total length of $\tau$ spent in such cycles and in the simple paths of $\tau$ (of total length at most $b$) is $X\le b(b^5\cdot \P-1)+b=b^6\cdot \P$.
		The effect accumulated by the transitions in $X$ is at most $X$, if all the relevant transitions are $+1$.
		
		The remaining part of the path, whose length is denoted by $Y=\ell-X$, is spent in cycles with $\es_j$ for $j\le i$. Note that its effect is at most $(v+X-\Prev{\cT})$, and therefore its length $Y \leq \frac{-1}{\es_{i}}(v+X-\Prev{\cT})$.
		
		Observe that any actual path $\tau$ that satisfies the required constraints cannot be shorter than the above theoretical path in which there are $X$ transitions of effect $(+1)$.
		
		Therefore, we have $\ell\leq X + Y \leq b^6\cdot \P + \frac{-1}{\es_{i}}(v+b^6\cdot \P-\Prev{\cT}) = 
		\frac{-1}{\es_{i}}(v-\Prev{\cT}) + (\frac{-1}{\es_{i}}+1) b^3\cdot \P \leq  \frac{-1}{\es_{i}}(v-\Prev{\cT}) + (b^7+b^6)\cdot\P$.
		
		Recall, however, that $\ell\ge \segStart_{i}(v)+\frac{\sT}{2}=
		(\frac{-1}{\es_{i}}(v-\Prev{\cT})-b^8\cdot \P) + (\frac{1}{2}b^9\cdot\P)$, leading to a contradiction, as $\frac{b^9}{2}>b^6+b^7+b^8$, for $b>2$.

		\item
		Consider a basic path $\tau$ of length $\ell< \segStart_{i+1}(v)$, such that the counter values of $\tau$ reach $\Prev{\cT}$. Without loss of generality, we can assume that $\ell$ is the first time when $\tau$ reaches $\Prev{\cT}$ (otherwise we will look at the prefix of $\tau$ that satisfies this).
		Assume by way of contradiction that $\tau$ has less than $b^5\cdot \P$ cycles with $\es_j$ for $j< i$, we show that $\ell$ must fall after $\segStart_{i+1}(v)$.
		
		Since there are at most $b^5\cdot \P-1$ cycles with $\es_j$ for $j\le  i$, each of length at most $b$, then the total length of $\tau$ spent in such cycles and in the simple paths of $\tau$ (of total length at most $b$) is $X\le b(b^5\cdot \P-1)+b=b^6\cdot \P$.
		The effect accumulated by the transitions in $X$ is at least $-X$, if all the relevant transitions are $-1$.
		
		The remaining part of the path, whose length is denoted by $Y=\ell-X$, is spent in cycles with $\es_j$ for $j> i$.
		Note that its effect is at least $(v-X-\Prev{\cT})$, and therefore its length $Y \geq \frac{-1}{\es_{i+1}}(v-X-\Prev{\cT})$.
		
		Observe that any actual path $\tau$ that satisfies the required constraints cannot be shorter than the above theoretical path in which there are $X$ transitions of effect $(-1)$.
		
		Therefore, we have $\ell\geq X + Y \geq b^6\cdot \P + \frac{-1}{\es_{i+1}}(v-b^6\cdot \P-\Prev{\cT}) = 
		\frac{-1}{\es_{i+1}}(v-\Prev{\cT}) + (\frac{-1}{\es_{i+1}}-1) b^6\cdot \P \geq  \frac{-1}{\es_{i+1}}(v-\Prev{\cT}) - b^7\cdot\P$.
		
		Recall, however, that $\ell< \segStart_{i+1}(v)=\frac{-1}{\es_{i+1}}(v-\Prev{\cT})-b^8\cdot \P$, leading to a contradiction.	
	\end{enumerate}
\end{proof}
\subsection{Remaining Cases in the Proof of \cref{cl:Periodicity}}
We present remaining arguments for the proof. 
\subsubsection{Proof of \cref{cl:Periodicity}.\cref{item:SegmentPeriodicityLengthening}}

Consider a length $\ell\notin \core(v)$, and let $\tau$ be a basic path of length $\ell-\P$ such that 
$(s,v)\leads{\ell-\P}(e,u)$. 
We construct from $\tau$ a path $\tau'$ for  $(s,v)\leads{\ell}(e,u')$, such that $u\equiv u'$.

This part of the proof, showing that a path can be properly lengthened by $\P$, is partially analogous to part \cref{item:SegmentPeriodicityShortening} of the proof, which shows that a path can be shortened by $\P$. 
We split the proof to the same cases as in part \cref{item:SegmentPeriodicityShortening}, and whenever the proofs are similar enough, we refer to the corresponding case in part \cref{item:SegmentPeriodicityShortening}.
\paragraph*{Case 1b.1: The counter values in $\tau$ stay above $\Prev{\cT}$}

Since there is no position in $\tau$ with counter value $\Prev{\cT}$, then in particular $\tau$ has no zero-transitions. We split to subcases.

\begin{itemize}
	\item 	If there is a non-negative cycle $c$ in $\tau$, we can obtain $\tau'$ by adding $\frac{\P}{|c|}$ copies of $c$, having that the counter values in $\tau'$ are at least as high as the corresponding ones in $\tau$.
	Observe that $\tau'$ is of length $\ell+\P$ from $(s,v)$ to $(e,u')$ with $u'=u+\effect{c}\frac{\P}{|c|}$. Since we have $u'\ge u\ge \Prev{\cT}$, then $u\equiv u'$.
	\item 	If there are in $\tau$ two cycles of different slopes, such that each of them repeats at least $b\cdot\P$ times, we can use \cref{cl:CycleCombinationForElr} to properly lengthen $\tau$ into $\tau'$.
	\item 	Otherwise, we are in the case that only a single cycle $c$ with some negative slope $\es$ repeats at least $b\cdot\P$ times. Since $\ell-\P>\segStart_i(v)+\frac{\sT}{2}$, by \cref{cl:SegmentsCycles}.\cref{Item:SlowCycles}, we have $\es>\es_i$.
	We can thus add $\frac{\P}{|c|}$ copies of $c$, guaranteed that the counter values in $\tau'$ remain above $\Prev{\cT}$.
	Indeed, the maximal drop (i.e., the negation of the minimum cumulative counter effect in the path) $D$ of $\tau$ is $X+Y$, where $X$ stands for the repetitions of $c$, and $Y$ for the rest. We have $X\leq -\es\cdot(\ell-\P)$ (notice that $\es$ is negative) and $Y\leq b+ (b-1)b^2\P$, considering the worst case in which all transitions out of cycles, as well as in the cycles that are not $c$ are of effect $(-1)$. Thus, $D \leq -\es\cdot(\ell-\P) + b^3\P$.
	
	As $\ell-\P<\ell<\segStart_{i+1}(v)=\frac{(v-\Prev{\cT})}{-\es_{i+1}}-b^8\cdot\P$, we have
	$D \leq \frac{\es}{\es_{i+1}}(v-\Prev{\cT}) + b^3\cdot\P - (-\es)b^8\cdot\P$.
	Since $\es>\es_i$, it follows that $\es\geq\es_{i+1}$, thus (as both $\es$ and $\es_{i+1}$ are negative), we have
	$D \leq (v-\Prev{\cT}) + b^3\cdot\P - (-\es)b^8\cdot\P$.
	
	Recall that we aim to show that $v-D$ can ``survive'' a drop of up to $\P$ additional copies of $c$ (i.e., at most $b\cdot\P$), while keeping the counter value above $\Prev{\cT}$. Thus, we need to show that $v - ((v-\Prev{\cT}) + b^3\cdot\P - (-\es)b^8\P) - b\cdot\P > \Prev{\cT}$.
	It is left to show that $ (-\es)b^8\cdot\P > (b^3+b)\cdot\P$, which obviously holds, as $-\es\geq\frac{1}{b}$.	
\end{itemize}

\paragraph*{Case 1b.2: $\tau$ reaches counter value $\Prev{\cT}$}

Let $0\le z_f\le z_u\le \ell-\P$ be the first and ultimate positions in $\tau$ where the counter value is exactly $\Prev{\cT}$. 
We split $\tau$ into three parts: $\tau_1 = \tau[0~..~z_f), \tau_2=\tau[z_f~..~z_u), \tau_3=\tau[z_u~..~\ell-\P]$ as in the analogous case in~\cref{item:SegmentPeriodicityShortening}.

\begin{enumerate}
	
	\item \emph{The middle part  $\tau_2=\tau[z_f~..~z_u]$ or last part $\tau_3=\tau[z_u~..~\ell-\P]$ are of length at least $b^8\cP$.}
	These cases are analogous to their counterparts \cref{item:LongMiddlePart} and \cref{item:LongLastPart} of \cref{item:SegmentPeriodicityShortening}.2, by adding cycle repetitions instead of removing them.
	
	%
	%
	%
	
	\item {\it Only the first part $\tau_1=\tau[0~..~z_f)$ of length at least $b^8\cP$}\label{item:LongFirstPart1b}
	
	
	If any of the other parts is long, we lengthen them.
	Otherwise, their combined length is less than $2b^8\cdot\P<\frac{\sT}{2}$, implying that the first part $\tau_1$ is longer than $\segStart_i(v)+\frac{\sT}{2}$.
	
	Hence, by \cref{cl:SegmentsCycles}, there are `fast' and `slow' cycles $c_f$ and $c_s$, respectively, of slopes $\es_f<\es_s$, such that each of them repeats at least $b^4\cdot\P$ times in $\tau_1$.
	
	Thus, by \cref{cl:CycleCombinationForLength}, we can add and/or remove some repetitions of $c_f$ and $c_s$, such that $\tau_1$ is longer by exactly $\P$.
	Yet, we should ensure that $\tau'$ is  valid, in the sense that its corresponding first part $\tau'_1$ cannot get the counter value to $0$.		
	We show it by cases:
	\begin{itemize}
		\item If $c_f$ or $c_s$ are balanced, we can lengthen  without changing the remaining counters.
		\item If there is a positive cycle $c_{+}$ that repeats at least $2b^2\cdot\P$ times, then the counter value climbs by at least $2b^2\cdot\P$ from its value $v_x$ at position $x$ where $c_{+}$ starts and the position $y$ where its repetitions end. 
		As the counter gets down to $\Prev{\cT}<\P$ at the end of $\tau_1$, and $\Prev{\cT}<P$ (\cref{cl:PeriodBiggerThanPreviousThreshold}), there must be a negative cycle $c_{-}$ that repeats at least $b\cdot\P$ times between position $y$ and the first position after $y$ that has the counter value $v_x+b^2\cdot\P$. Hence, we can add repetitions of $c_{+}$ and $c_{-}$, as per \cref{cl:CycleCombinationForLength}, ensuring that the only affected values are above $v_x$.
		\item Otherwise, both $c_f$ and $c_s$ are negative, implying that we add some repetitions of $c_s$ and remove some repetitions of $c_f$. We further split into two subcases:
		\begin{itemize}
			\item If $c_f$ appears before $c_s$ then there is no problem, as the only change of values will be their increase. 
			\item If $c_s$ appears first, ending at some position $x$, while $c_f$ starts at some position $y$, then a-priori it might be that repeating $c_s$ up to $b\cdot\P$ more times, as per \cref{cl:CycleCombinationForLength}, will take the counter value to $0$. 
			
			Yet, observe that since there are at most $b-2$ positive cycles, and each of them can repeat at most $2b^2\cdot\P-1$ times, the counter value $v_x$ at position $x$, and along the way until position $y$, is at least $v_y-(b-2)2b^2\cdot\P$, where $v_y$ is the counter value at position $y$. As $c_f$ repeats at least $b^4\cdot\P$ times, we have $v_y\geq b^4\cdot\P$. Thus $v_x \geq b^4\cdot\P - 2(b-2)b^2\cdot\P > b^2\P$. Hence, repeating $c_s$ up to $b\cdot\P$ more times at position $x$ cannot take the counter value to $0$, until position $y$, as required.
		\end{itemize}
	\end{itemize}	
\end{enumerate}		

\subsubsection*{Continuation of the Proof of \cref{cl:Periodicity}.\cref{item:CoreShiftShortening}}
	We detail the subcases of Case 2a.2.2 that were not detailed in the main text:
	\begin{itemize}
		\item If there is such a relatively slow cycle $c_s$ that ends in a position $x$ before a position $y$ in which the first really fast cycle $c_f$ starts, then:
		If $\es_s=s_i$, we can simply remove $\frac{\P}{-\es_s \cdot |c_s|}$ repetitions of $c_s$, getting the desired counter and lengths changes.
		
		Otherwise, namely when $\es_s>s_i$, we claim that the counter value between positions $x$ and $y$ is high enough, allowing to remove some repetitions of $c_s$ and $c_f$, as per  \cref{cl:CycleCombinationForElr}. 
		As $c_s$ might be positive, removal of its repetitions might decrease the counter; the worst case is when all its transitions are of effect $(+1)$, and we shorten the path the most, which is bounded by $b\cdot\P$. To ensure that the counter remains above the value $P$ after the removal and until position $y$, we thus need it to be at least $b\cdot\P+\P$ at all positions until $y$.
		
		Indeed, the path has only cycles of slope at least $\es_{i+1}$ until position $y<\segStart_{i}+\sT$. Therefore, its drop $D$ is at most $b$, for the non-cyclic part, plus $(-\es_{i+1})(\segStart_{i}(v+\P)+\sT)$ for the cyclic part. (Notice that $\es_{i+1}$ is negative.)
		Therefore, $D \leq b + (-\es_{i+1})(\segStart_{i}(v+\P)+\sT)$. Taking the value of $\segStart_{i}(v+\P)$, we have 
		$D \leq b + (-\es_{i+1})( \frac{-1}{\es_{i}}(v+\P-\Prev{\cT})-b^8\cdot \P+\sT)
		= b + \frac{\es_{i+1}}{\es_{i}}(v+\P-\Prev{\cT})+(\es_{i+1})(b^8\cdot \P+\sT)$.
		Since $\es_{i+1}\leq\frac{1}{b}$ and $\frac{\es_j}{\es_i}\leq1-\frac{1}{b^2}$, we have 
		$D\leq b + (1-\frac{1}{b^2})(v+\P-\Prev{\cT})+\frac{1}{b}(b^8\cdot \P+\sT)
		\leq b + (1-\frac{1}{b^2})(v+\P-\Prev{\cT})+(b^7\cdot \P+b^8\cdot\P)$.
		
		Now, as the counter starts with value $v+\P$, we need to show that $E = v+\P - D -(b^2+1)\P \geq 0$.
		Indeed, $E \geq \frac{1}{b^2}(v+\P) - b^7\cdot \P- b^8\cdot\P-(b^2+1)\P$.
		As $v\geq\cT=b^{11}\cdot\P$, we have $E \geq (b^9- b^7 - b^8-b^2-1)\P>0$, as required.
		
		\item Otherwise, namely when there is no relatively slow cycle $c_s$ that repeats at least $b^4\cdot\P$ times before the position $y$ in which the first really fast cycle $c_f$ starts:
		
		Let $z$ be the position after $y$, in which the path has the smallest counter value $v_z$ before a relatively slow cycle that repeats $b^4\cdot\P$ times.

		If $v_z>\P$, we can remove some repetitions of the really fast and relatively slow cycles, as per \cref{cl:CycleCombinationForElr}; We only remove repetitions of the fast decreasing cycle, so as a result the counter will only grow, and if it is guaranteed to be above $\P$ when starting the path from counter value $v+\P$, it will not reach $0$ when starting from counter value $v$.
		
		Otherwise, namely when $v_z\leq\P$, we can remove $\frac{\P}{-\es_f \cdot |c_f|}$ repetitions of the really fast cycle $c_f$, guaranteeing that $\tau'_1$ remains above $0$ until position $y$.
		
		Observe, however, that the resulting path $\tau'$ will be longer than $\ell$. 
		Nevertheless, we claim that in this case, the portion of $\tau_1$ from position $z$ until its last position with counter value $v_z$ can be further shorten without changing the path's effect. Indeed, there is a cycle that repeats at least $b^4\cdot\P$ times in this part, having an absolute effect of at least $b^3\cdot\P$, implying that this part of $\tau_1$ climbs up to a value of at least  $b^3\cdot\P$, and down again, allowing for the removal of positive and negative cycles, as per \cref{cl:CycleCombinationForLength}, analogously to \cref{item:LongMiddlePart} in the proof of \cref{cl:Periodicity}.\cref{item:SegmentPeriodicityShortening}.2.
	\end{itemize}

\subsubsection{Proof of \cref{cl:Periodicity}.\cref{item:CoreShiftLengthening}}
We remark that Segment $0$ is identical in the shortening and lengthening argument, and this is a repetition of the main text. 
For a path of length $\ell$, we have in Segment $0$ that $\shift(\ell)=\ell$, and indeed a path from $v$ is valid from $v+\P$ and vise versa, as they do not hit $\Prev{\cT}$: Their maximal drop is $\sT$, while $v\geq\cT>\P+\sT>\Prev{\cT} + \sT$.

We turn to the $i$th segment, for $i\geq 1$.
Consider a basic path $\tau$ for $(s,v)\leads{\ell}(e,u)$. Recall that  $\shift(\ell)=\ell+\frac{\P}{-\es_i} \in [\segStart_i(v+\P) ~..~ \segStart_i(v+\P)+\sT)]$. 
We construct from $\tau$ a path $\tau'$ for  $(s,v+\P)\leads{\shift(\ell)}(e,u')$, such that $u\equiv u'$, along the following cases.

\paragraph*{Case 2b.1: The counter values in $\tau$ stay above $\Prev{\cT}$}

Since there is no position in $\tau$ with counter value $\Prev{\cT}$, then in particular $\tau$ has no zero-transitions. 
Since $\ell>\segStart_{i-1}(v)+\frac{\sT}{2}$, by \cref{cl:SegmentsCycles}.\cref{Item:SlowCycles}, there is a cycle $c$ in $\tau$ of slope $\es>\es_{i-1}$, namely of slope $\es\geq\es_i$.
We can thus get the required path $\tau'$, by adding $\frac{\P}{-\es \cdot |c|}$ repetitions of $c$. Indeed, the counter value of $\tau'$ at the position that $c$ starts will be bigger by $\P$ than the counter value of $\tau$ at that  position, while the counter of $\tau'$ after the repetitions of $c$ will be at least as high as the counter value of $\tau$ at the corresponding position.

\paragraph*{Case 2b.2: $\tau$ reaches counter value $\Prev{\cT}$}

Let $0\le z_f\le z_u\le \ell$ be the first and ultimate positions in $\tau$ where the counter value is exactly $\Prev{\cT}$. We split $\tau$ into three parts: $\tau_1 = \tau[0~..~z_f), \tau_2=\tau[z_f~..~z_u), \tau_3=\tau[z_u~..~\ell]$ (it could be that $z_f=z_u$, in which case the middle part is empty).

In order to accommodate with possible \ZTs, we should lengthen $\tau_1$, such that the resulting first part $\tau'_1$ of the new path $\tau'$, which starts with counter value $v+\P$, will also end with counter value exactly $\Prev{\cT}$. 
Since $\tau_1$ reaches $\Prev{\cT}$ and it is shorter than $\segStart_{i+1}(v)$, it has by \cref{cl:SegmentsCycles}.\cref{Item:FastCycles} at least $b^4\cdot\P$ repetitions of a `fast' cycle $c_f$ of slope $\es_f\leq \es_i$. 

\begin{enumerate}
	\item
	If $\es_f=\es_i$ or $\tau_2$ or $\tau_3$ are of length at least $b^5\cdot\P$, we can add $\frac{\P}{-\es_f \cdot |c_f|}$ repetitions of $c_f$ in $\tau_1$. Note that the resulting first part $\tau'_1$ of the new path $\tau'$ indeed ends with counter value $\Prev{\cT}$. However, while when  $\es_f=\es_i$  the resulting length of $\tau'$ will be $\shift(\ell)$, as required, in the case that $\es_f<\es_i$, we have that the resulting path $\tau'$ will be shorter than $\shift(\ell)$. Nevertheless, in this case, as $\tau_2$ or $\tau_3$ are of length at least $b^5\cdot\P$, we can further lengthen $\tau_2$ or $\tau_3$ without changing their effect, as per \cref{cl:CycleCombinationForLength}, analogously to \cref{item:LongMiddlePart} or \cref{item:LongLastPart}, respectively, in the proof of \cref{cl:Periodicity}.\cref{item:SegmentPeriodicityShortening}.2.
	
	\item
	Otherwise, we are in the case that $\tau_1$ has a `really fast' cycle of slope $\es_f< \es_i$ that repeats at least $b^4\cdot\P$ times, and both $\tau_2$ or $\tau_3$ are of length less than $b^5\cdot\P$.
	We claim that in this case $\tau_1$ must also have $b^4\cdot\P$ repetitions of a `relatively slow' cycle $c_s$ of slope $\es_s\geq \es_i$. The argument for that is the same as in \cref{Item:FastSlowInPartOne2a} of  part \cref{item:CoreShiftShortening} of the proof.
	As the case of $\es_s= \es_i$ is handled in the previous subcase, we may further assume that $\es_s> \es_i$.
	We may thus add repetitions of $c_f$ and $c_s$, as per \cref{cl:CycleCombinationForElr}, lengthening $\tau_1$ by exactly $\shift(\ell)-\ell$, while also ensuring that it ends with counter value exactly $\Prev{\cT}$.
	
\end{enumerate}		

\end{document}